\documentclass[twoside]{article}

\usepackage[round]{natbib}

\usepackage{soul}
\usepackage{url}
\usepackage[hidelinks]{hyperref}
\usepackage[utf8]{inputenc}
\usepackage[small]{caption}
\usepackage{graphicx}
\usepackage{amsmath,amssymb,amsthm}
\usepackage{booktabs}
\usepackage{mathtools}
\usepackage{latexsym}
\usepackage{bm}
\usepackage{setspace}
\usepackage{enumerate}
\usepackage{cases}
\usepackage{mathrsfs}
\usepackage{subcaption}
\usepackage{comment}
\usepackage{url}
\usepackage{array}

\usepackage{wrapfig}
\usepackage{tikz}
\usepackage{tikz-3dplot}
\usetikzlibrary{positioning}
\usepackage[normalem]{ulem}

\usepackage{multirow}

\usepackage{cleveref}
\usepackage{autonum}
\allowdisplaybreaks

\usepackage{algorithm} 
\usepackage{algorithmicx}
\usepackage[noend]{algpseudocode}

\def\bma{\begin{bmatrix}}
\def\ema{\end{bmatrix}}
\def\bpma{\begin{pmatrix}}
\def\epma{\end{pmatrix}}

\def\P{{\mathrm P}}

\def\Z{\mathbb Z}
\def\R{\mathbb R}
\def\1{\mathbf 1}
\def\0{\mathbf 0}

\newcommand{\ba}{\begin{align}}
\newcommand{\ea}{\end{align}}

\newcommand{\sign}{\mathop{\mathrm{sign}}\nolimits}

\newcommand{\prox}{\mathop{\mathrm{prox}}\nolimits}

\newcommand{\argmin}{\mathop{\mathrm{argmin}}\limits}

\newcommand{\innprod}[2]{\langle #1, #2 \rangle}

\newcommand{\ceil}[1]{\left\lceil {#1} \right\rceil}
\newcommand{\floor}[1]{\left\lfloor {#1} \right\rfloor}

\newcommand{\supp}{\mathop{\mathrm{supp}}\limits}

\DeclareMathOperator*{\minimize}{minimize}
\DeclareMathOperator*{\maximize}{maximize}
\DeclareMathOperator{\subto}{subject\ to}
\newcommand{\relmid}[1]{\mathrel{#1|}}

\newtheorem{thm}{Theorem}% Theorem
\theoremstyle{definition}
\renewcommand{\P}{P}
\newcommand{\F}{F}

\newcommand{\minheap}{\text{$\mathrm{MinHeap}$}}
\newcommand{\pop}{\mathrm{pop}}
\newcommand{\push}{\mathrm{push}}

\newcommand{\solver}{\text{${\tt SubtreeSolver}$}}
\newcommand{\desc}{\mathrm{desc}}

\newcommand{\dset}{{[d]}}
\newcommand{\V}{V}

\newcommand{\myle}{\le}
\newcommand{\myg}{>}
\newcommand{\myl}{<}

\newcommand{\low}[1]{{\tt Low}_{#1}}
\newcommand{\sol}[1]{{\tt Sol}_{#1}}

\newcommand{\dmax}{D_{\max}}
\newcommand{\pmin}{\P_{\min}}
\newcommand{\solmin}{\sol{{\min}}}

\newcommand{\U}{U}
\newcommand{\efdom}{{\mathcal{F}}}
\newcommand{\PF}{\mathcal{P}_\efdom}

\newcommand{\Lhuber}{L_{\text{Huber}}}
\newcommand{\xtrue}{x_{\text{true}}}
\newcommand{\xnoise}{x_{\text{noise}}}
\newcommand{\bnoise}{b_{\text{noise}}}
\newcommand{\randn}{\mathcal{N}}

\newcommand{\Tl}{{\mathcal{T}}}
\newcommand{\topknorm}[2]{\|{#1}\|_{{#2},2}}
\newcommand{\xtl}{\tilde{x}}

\newcommand{\Thetarm}{{\mathrm{\Theta}}}
\usepackage[margin=.8in]{geometry}

\begin{document}
%	\runningtitle{}
	
	% If your paper is accepted and the title of your paper is very long,
	% the style will print as headings an error message. Use the following
	% command to supply a shorter title of your paper so that it can be
	% used as headings.
	%
	%\runningtitle{I use this title instead because the last one was very long}
	
	% If your paper is accepted and the number of authors is large, the
	% style will print as headings an error message. Use the following
	% command to supply a shorter version of the authors names so that
	% they can be used as headings (for example, use only the surnames)
	%
	%\runningauthor{Surname 1, Surname 2, Surname 3, ...., Surname n}
	
%	\twocolumn
	
	\title{
		Best-first Search Algorithm for Non-convex Sparse Minimization
	}
	
	\author{Shinsaku Sakaue 
		\\
		NTT Communication Science Laboratories
		\and 
		Naoki Marumo
		\\
		NTT Communication Science Laboratories}
	
	\maketitle
	
	\begin{abstract}
	Non-convex sparse minimization (NSM), or $\ell_0$-constrained minimization of convex loss functions, is an important optimization problem that has many machine learning applications. NSM is generally NP-hard, and so to exactly solve NSM is almost impossible in polynomial time. As regards the case of quadratic objective functions, exact algorithms based on quadratic mixed-integer programming (MIP) have been studied, but no existing exact methods can handle more general objective functions including Huber and logistic losses; this is unfortunate since those functions are prevalent in practice. In this paper, we consider NSM with $\ell_2$-regularized convex objective functions and develop an algorithm by leveraging the efficiency of best-first search (BFS). Our BFS can compute solutions with objective errors at most $\Delta\ge0$, where $\Delta$ is a controllable hyper-parameter that balances the trade-off between the guarantee of objective errors and computation cost. Experiments demonstrate that our BFS is useful for solving moderate-size NSM instances with non-quadratic objectives and that BFS is also faster than the MIP-based method when applied to quadratic objectives.
	\end{abstract}

\section{INTRODUCTION}

We consider non-convex sparse minimization (NSM) problems formulated as follows: 
\begin{align}
\minimize_{x\in\R^d}\ 
\P(x)  
& & 
\subto\ 
\|x\|_0\le k, 
\label{prob:nsm}
\end{align}
where 
$\|x\|_0$ is the number of non-zeros in $x$ 
and $k\in\Z_{>0}$ is a sparsity parameter. 
We assume $\P:\R^d\to\R$ 
to be a convex function  
(e.g., quadratic, Huber, and logistic losses) 
with $\ell_2$-regularization as detailed in \Cref{section:oracle}.  
The feasible region of NSM is non-convex 
due to the $\ell_0$-constraint, 
hence NSM generally is NP-hard \citep{natarajan1995sparse}. 
NSM is important since it arises in many real-world scenarios 
including 
feature selection~\citep{hocking1967selection} 
and 
robust regression~\citep{bhatia2017consistent}. 
Due to the NP-hardness,
most studies on NSM have been devoted to 
{\it inexact} polynomial-time algorithms, 
which are not guaranteed to find optimal solutions: 
E.g., orthogonal matching pursuit (OMP)~\citep{pati1993orthogonal,elenberg2018restricted}, 
iterative hard thresholding (IHT)~\citep{blumensath2009iterative,jain2014iterative}, 
and  
hard thresholding pursuit (HTP)~\citep{foucart2011hard,yuan2016exact}. 

When solving moderate-size NSM instances 
to make a vital decision, 
the demand for {\it exact} algorithms increases.\footnote{We say an optimization algorithm is {exact} 
	if it is guaranteed to achieve optimal objective values, 
	where we admit small errors 
	such as those arising from the machine epsilon.} 
Moreover, to exactly solve NSM  
is useful for revealing the performance and limitation of the inexact algorithms, 
which helps advance research into NSM. 
These facts motivate us to develop efficient exact algorithms for NSM. 
The main difficulty of exactly solving NSM stems from 
the fact that there are $\Thetarm(d^k)$ non-zero 
patters, or {\it supports}; 
to examine them one by one is prohibitively costly. 
For the case where $P(x)$ is quadratic, 
i.e., $P(x)=\frac{1}{2n}\|b-Ax\|_2^2$ where $b\in\R^n$ and $A\in\R^{n\times d}$, 
\citet{bertsimas2016best} have developed  
an exact algorithm based on quadratic mixed-integer programming (MIP). 
Since MIP solvers (e.g., Gubori and CPLEX) 
employ sophisticated search strategies such as the branch-and-bound (BB) method, 
their method works far more efficiently than the exhaustive search.  

In practice, non-quadratic objective functions are very common;   
e.g., if observation vector $b$ includes outliers, 
we use the Huber loss 
function to alleviate the effect of outliers. 
The MIP-based method 
is inapplicable to such NSM instances since they 
cannot be formulated as quadratic MIP. 
To the best of our knowledge, 
no existing exact NSM algorithms can handle  
general convex functions such as Huber and logistic losses.

In this paper, we develop the first exact NSM algorithm 
that can deal with $\ell_2$-regularized convex objective functions. 
Our algorithm searches for an optimal support based on the best-first search (BFS) \citep{pearl1984heuristics}, 
which is a powerful search strategy including the A* search \citep{hart1968formal}. Since there are few studies on NSM in the field of search algorithms, 
a BFS framework and its several key components for NSM 
remain to be developed; the most important is 
an {\it admissible heuristic}, which appropriately prioritizes candidate supports. 
We develop such a prioritization method, 
called \solver, 
and some additional techniques by utilizing 
the latest studies on continuous optimization methods \citep{liu2017dual,malitsky2018first}. 
Although the two main building blocks of our algorithm, 
BFS and continuous optimization methods, 
are well studied, to develop 
an NSM algorithm by utilizing them 
requires careful discussion as in \Cref{section:bfs,section:oracle}. 	
Below we detail our contributions: 

\begin{itemize}
	\item
	In \Cref{section:bfs}, 
	assuming that \solver\ is available, 
	we show how to search for an optimal support via BFS. 
	Our BFS outputs solutions with objective errors at most $\Delta\ge0$, 
	where $\Delta$ 
	is an input that controls the trade-off between the 
	accuracy and computation cost;  
	BFS is exact if $\Delta=0$, 
	and BFS empirically speeds up as $\Delta$ increases.   
	
	\item 
	In \Cref{section:oracle}, 
	we develop \solver\ that works with 
	$\ell_2$-regularized convex functions. 
	We also develop two techniques that accelerate BFS: 
	Pruning of redundant search space and 
	warm-starting of \solver. 
	Although pruning is common in the area of 
	heuristic search, 
	how to apply it to NSM is non-trivial. 
	Experiments in \Cref{app:ablation} 
	confirm that BFS greatly speeds up thanks to 
	the combination of the two 
	techniques.
	
	\item 
	In \Cref{section:experiments}, 
	we validate our BFS via experiments.  
	We confirm that BFS can exactly solve NSM instances 
	with non-quadratic objectives, 
	which inexact algorithms fail to solve exactly;  
	this implies optimal solutions of some NSM instances 
	cannot be obtained with other methods than BFS. 
	We also demonstrate that to exactly solve NSM is beneficial in terms of support recovery. 
	Experiments with quadratic objectives 
	show that BFS is more efficient than the MIP-based method 
	that uses the latest commercial solver, Gurobi 8.1.0.
\end{itemize}

\subsection{Related Work}\label{subsec:related}	 
We remark that BFS is different from BB-style methods, 
which MIP solvers often employ.  
While BB needs to examine or prune 
every possible support to guarantee the optimality of output solutions, 
BFS can output optimal solutions without examining the whole search space. 
This property is obtained from the admissibility of heuristics; 
in our case, it holds thanks to the design of \solver. 
In \Cref{subsec:real}, 
we confirm that our BFS can run faster than the MIP-based method 

Our work is also different from 
previous studies that consider similar settings. 
\citet{huang2018constructive} 
studied the $\ell_0$-{\it penalized} 
minimization of {\it quadratic} objectives, 
while 
we consider $\ell_0$-{\it constrained} minimization 
with {\it $\ell_2$-regularized convex} objectives. 
MIP approaches to other penalized settings 
are studied in \citep{miyashiro2015mixed,sato2016feature}, 
but the $\ell_0$-constrained setting is not considered. 
\citet{bourguignon2016exact} 
studied 
{\it Big-$M$}-based MIP formulations 
of sparse optimization whose objective functions are 
given by the $\ell_p$-norm. 
Unlike our BFS and the MIP approach \citep{bertsimas2016best}, 
their method requires the assumption that 
no entry of an optimal solution is larger in absolute value than predetermined $M>0$. 
Furthermore, 
our BFS accepts objective functions 
other than those written with the $\ell_p$-norm. 
\citet{bertsimas2017sparse} in a preprint have recently proposed 
a cutting-plane-based MIP approach; 
as with the previous MIP approach \citep{bertsimas2016best}, 
however, it is accepts only quadratic objectives. 
\citet{karahanoglu2012astar} 
and 
\citet{arai2015optimal}
proposed 
A* algorithms for {compressed sensing} 
and {column subset selection}, respectively; 
our problem setting and design of the heuristic are different 
from those in their works. 

%	\rr{trimming L1: \citep{yun2019trimming}}	
%	\memo{related?: The Impact of Regularization on High-dimensional Logistic
%		Regression}

\section{BFS FRAMEWORK}\label{section:bfs}
We first define the {\it state-space tree}, 
on which we search for an optimal support, 
and we then show how to perform BFS on the tree. 
The state-space tree is often used in reverse search~\citep{avis1996reverse}, 
and similar notions are used for 
{submodular} maximization~\citep{chen2015filtered,sakaue2018accelerated}.

Let $\dset\coloneqq\{1,\dots,d\}$ be the index set of $x\in\R^d$; 
we take  
the elements in $\dset$ to be totally ordered as $1\myl2\myl\cdots\myl d$.      
Given any $S\subseteq \dset$, we let $\max S\in \dset$ denote 
the largest element in $S$.  
For any $x\in\R^d$, 
we let 
$\supp(x)\subseteq \dset$	
denote 
the {support} of $x$, 
which is the set of indices 
corresponding to the non-zero entries of $x$. 
We let $x^*$ be an optimal solution to problem~\eqref{prob:nsm}.

\subsection{State-space Tree}
We define the state-space tree $G=(\V,E)$ as follows. 
The node set is given by 
\begin{equation}
\V\coloneqq\{S\subseteq \dset \relmid{} |S|\le k \text{ and } k - |S| \le d-\max S \},
\end{equation}
and 
$S,T\in \V$ are connected by directed edge $(S,T)\in E$ 
iff $S=T\backslash \{ \max T\}$.  
%, which are candidates for $\supp(x^*)$. 
Roughly speaking, 
each $S\in\V$ 
represents an index subset 
such that the corresponding entries 
are allowed to be non-zero.  
Since $\V$ includes 
all $S$ of size $k$, 
any support of size at most $k$ is included in 
some $S\in\V$.  
Let $\desc(S)\subseteq \V$ denote 
a node subset that comprises 
$S$ and all its descendants. 
\Cref{fig:tree} shows an example state-space tree. 	
While 
the size of the tree, $|\V|$, is exponential in $k$, 
BFS typically requires only a very small fraction of $\V$ on-demand, 
which we will experimentally confirm in \Cref{subsubsec:cost}.  

%	which is prohibitively large in practice. 
%	Fortunately, however, the nodes of 
%	the tree are examined on-demand with BFS and, 
%	typically, 
%	only a very small fraction of $\V$ is examined throughout 
%	the search process. 

\begin{figure}[tb]
	\centering
	\hspace{-3cm}
	\begin{tabular}{p{0.25\textwidth}}
		\centering
		\begin{tikzpicture}[every node/.style={}]
		%		\fill[fill=red, opacity=.2, rounded corners]
		%		(-.8,-.8) -- 
		%		(.65,-2.0)  -- 
		%		(.65,-3.5)  -- 
		%		(-5.0,-3.5) -- 
		%		(-1.7,-.8) -- 
		%		(-.8,-.8) -- 
		%		(.65,-2.0) ;
		\fill[fill=red, opacity=.3, rounded corners]
		(-.9,-.58) -- 
		(.62,-1.6)  -- 
		(.62,-2.85)  -- 
		(-4.5,-2.85) -- 
		(-4.5,-2.38) -- 
		(-1.6,-.58) -- 
		(-.9,-.58) -- 
		(-.14,-1.09) ;

		\node (r) {\uline{$\emptyset$}}; 
		\node[below left = 0.06\linewidth and 0.04\textwidth of r] (o) {\uline{$\{1\}$}};	
		\node[below left  = 0.06\linewidth and 0.02\textwidth of o] (ot) {\uline{$\{1,2\}$}};
		\node[below left = 0.06\linewidth and 0.001\textwidth of ot] (otth) {$\{1,2,3\}$};
		\node[below right = 0.06\linewidth and 0.001\textwidth of ot] (otf) {\uline{$\{1,2,4\}$}};
		\node[below right = 0.06\linewidth and 0.02\textwidth of o] (oth) {\uline{$\{1,3\}$}};
		\node[below = 0.06\linewidth of oth] (othf) {\uline{$\{1,3,4\}$}};
		\node[below right = 0.06\linewidth and 0.04\textwidth of r] (t) {$\{2\}$};
		\node[below = 0.06\linewidth of t] (tth) {$\{2,3\}$};
		\node[below = 0.06\linewidth of tth] (tthf) {$\{2,3,4\}$};
		
		\foreach \u / \v in {r/o,r/t,o/ot,o/oth,ot/otth,ot/otf,oth/othf,t/tth,tth/tthf}
		\draw[->, thick] (\u) -- (\v); 
		\end{tikzpicture}
	\end{tabular}
	\caption{State-space Tree 
		with $(d,k) = (4,3)$.  
		%			Note that
		%			the set of leaves forms $\{S\subseteq \dset \relmid{} |S|=k \}$.  
		%			which are candidates for the optimal support, 
		%			$\supp(x^*)$. 
		%			For any $S\in\V$, 
		%			$\desc(S)$ denotes a node subset that consists of 
		%			all descendants of $S$ and $S$ itself;  
		%			for example, $\desc(\{1\})$ includes 
		%			the nodes in the area colored in red.
		The nodes in the area shaded in red form $\desc(\{1\})$. 
		If $\supp(x^*)=\{1,4\}$, 
		$\V^*$ comprises the nodes 
		with the underlined labels.}
	\label{fig:tree}
\end{figure}
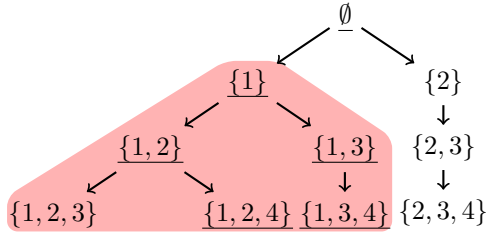

\subsection{Best-first Search on State-space Tree}\label{subsection:bfs} 
For any $S\in V$, 
we define 
\begin{equation}\label{def:u}
\U(S)\coloneqq\{ x\in\R^d \relmid{} 
\supp(x)\subseteq S^\prime,\ 
\exists S^\prime\in\desc(S)  \},  
\end{equation} 
which comprises all feasible solutions 
whose support is included in some $S^\prime\in\desc(S)$. 
Note that
\begin{equation}\label{eq:vstar}
\V^*\coloneqq\{S\in\V\relmid{} x^*\in\U(S)\}
\end{equation}
induces a subtree 
whose root is $\emptyset$ 
and leaves $S\in\V$ satisfy $|S|=k$ 
(see, \Cref{fig:tree}); 
we use this fact to prove the exactness of BFS.   
With BFS, 
starting from root $\emptyset\in\V$, 
we search for 
(a superset of) 
the optimal support, $\supp(x^*)$, 
in $G$ in a top-down 
manner.  
Since it is too expensive to examine all nodes in $G$, 
we reduce the search effort by 
appropriately prioritizing 
candidate nodes. 
We define set function $F:V\to\R$ as  
\begin{align}
\F(S)\coloneqq
\min_{x\in \R^d}
\{
P(x) \relmid{} 
x\in\U(S)
\}.
\label{prob:main}
\end{align} 
Note that  $\F(S)\ge\P(x^*)$ holds for any $S\in\V$ and that  
\begin{align}
\F(S)=\P(x^*) & & \text{$\forall S\in\V^*$}
\label{eq:f_eq_p}
\end{align} 
holds. 
Ideally, 
if we could use $F(S)$ as a priority value of $S\in V$, 
we could find $x^*$ without examining any redundant search space.  
However, to compute $F(S)$ is NP-hard in general.  
Thus we consider using an estimate of $F(S)$ 
that can be computed efficiently. 
As is known in the field of heuristic search~\citep{dechter1985generalized}, 
such estimates must satisfy certain conditions 
to guarantee the exactness of BFS.  
We assume that such estimates, as  well as candidate solutions, 
can be computed via \solver{} that satisfies the following requirements: 
$\solver(S)$ must output 
an estimate $\low{S}\in\R$ of $F(S)$ 
and 
a feasible solution $\sol{S}\in \R^d$ (i.e.,  $\|\sol{S}\|_0\le k$) 
that satisfy  
\begin{align}
&\low{S}\le \F(S) & &\text{$\forall S\in \V$}, 
\label{ineq:admissible} \\
&\P(\sol{S}) \le \low{S}+\Delta
& &\text{$\exists S\in \V^*$}. 
\label{eq:complete}
\end{align} 
\Cref{alg:bfs} describes BFS that uses \solver. 
Akin to the admissible heuristics of A* search, 
$\low{S}$ must lower bound $F(S)$ as in~\eqref{ineq:admissible}. 
Condition~\eqref{eq:complete} guarantees that 
BFS terminates in Step~\ref{step:return_sols}. 
In \Cref{section:oracle}, 
we develop $\solver$ satisfying \eqref{ineq:admissible} and \eqref{eq:complete}. 
As in \Cref{alg:bfs}, 
all examined $S\in\V$, 
as well as $\low{S}$ and $\sol{S}$,  
are 
maintained with \minheap, 
and they are prioritized with 
their $\low{S}$ values.   
In each iteration, $S$, the maintained node with the smallest 
$\low{S}$, 
is popped from \minheap, and then all its children are 
examined and pushed onto \minheap\ if not pruned in Step~\ref{step:detect_pruning}.  
%	These operations are continued until 
%	the condition in Step~\ref{step:termination} holds. 

\begin{algorithm}[t]
	{\fontsize{9pt}{10pt}\selectfont
		\caption{BFS for NSM}\label{alg:bfs}
		\begin{algorithmic}[1]
			\State $\low{\emptyset},\sol{\emptyset}
			\gets\solver(\emptyset)$
			\State $\minheap.\push(\low{\emptyset}, 
			\langle 
			\low{\emptyset}, \sol{\emptyset}, \emptyset
			\rangle
			)$
			\State 
			$\solmin\gets\sol{\emptyset}$
			and 
			$\pmin\gets P(\solmin)$
			\While{\minheap\ is not empty}
			\State 
			$\langle \low{S}, \sol{S}, S \rangle
			\gets \minheap.\pop()$
			\If{$\P(\sol{S})\le\low{S}+\Delta$} 
			\label{step:termination}
			\State \Return $\sol{S}$ 
			\label{step:return_sols}
			\EndIf
			\For{each $T$ such that $(S,T)\in E$}
			\State Start $\solver(T)$ to get $\low{T}$ and $\sol{T}$. 
			\If{$\low{T}>\pmin$ is detected} \Comment{Pruning}
			\label{step:detect_pruning}
			\State Force-quit $\solver(T)$. 
			\label{step:forcequit}
			\Else
			\State $\minheap.\push(\low{T}, 
			\langle \low{T}, \sol{T}, T \rangle )$
			\If{$P(\sol{T})<\pmin$}
			\State $\solmin\gets\sol{T}$ and $\pmin\gets P(\solmin)$
			\EndIf
			\EndIf
			\EndFor
			\EndWhile
		\end{algorithmic}
	}
\end{algorithm}

As in~\citep{hansen2007anytime}, 
we can use $\solmin$, the best current solution, 
to prune redundant search space.
More precisely, 
in Steps~\ref{step:detect_pruning} and~\ref{step:forcequit}, 
if $\low{T}>\pmin$ is detected while executing 
$\solver(T)$, 
we force-quit $\solver$ 
to reduce computation cost, 
and 
we never examine the redundant search space, 
$\desc(T)$.  
How to detect $\low{T}>\pmin$ depends on the design of $\solver$; 
we explain it in \Cref{subsec:acceleration}. 

Assuming that \solver\ is available, 
the exactness of \Cref{alg:bfs} can be proved as follows:
\begin{thm}\label{theorem:bfs} 
	For any $\Delta>0$, 
	%		if \solver\ satisfies~\eqref{ineq:admissible} 
	%		and~\eqref{eq:complete},  
	\Cref{alg:bfs} outputs a feasible $x$ that satisfies $P(x)\le P(x^*)+\Delta$. 
\end{thm}
%

%	From~\eqref{prob:main}, we have $\P(\sol{S})\ge\F(S)$, 
%	and hence the second inequality of~\eqref{ineq:admissible} naturally holds. 
%	Condition~\eqref{eq:complete} guarantees that 
%	BFS always returns a solution. 
%	Requirements \eqref{ineq:admissible} 
%	and~\eqref{eq:complete} 
%	imposed on \solver{} 
%	are associated with  
%	{\it soundness} and {\it completeness}, respectively, 
%	which are important notions studied in the field of 
%	search algorithms. 

%
\begin{proof}%[Proof of \Cref{theorem:bfs}]
	Since $\solmin$ is always feasible, 
	for any $T\in\V^*$,   
	\begin{align}
	P(\solmin) 
	&\ge P(x^*) & &\because\text{$\solmin$ is feasible}\\
	&= F(T) & &\because\text{$P(x^*)=F(T)$ from~\eqref{eq:f_eq_p}}\\
	&\ge \low{T} & &\because\text{$F(T)\ge\low{T}$ from~\eqref{ineq:admissible}}
	\end{align}
	holds. 		
	Therefore,  
	no $T\in\V^*$ is pruned in Step~\ref{step:detect_pruning}.  
	Furthermore, since 
	$\V^*$  %$\V^*=\{S\in\V\relmid{} x^*\in\U(S)\}$ 
	induces a subtree
	such that its root is $\emptyset$ 
	and its leaves $S\in\V^*$ satisfy $|S|=k$,  
	\minheap\ always maintains some $S^*\in\V^*$ until all nodes in $\V^*$ are popped. 
	Therefore, 
	thanks to~\eqref{eq:complete}, 
	BFS %always outputs a solution in Step~\ref{step:return_sols} (or BFS 
	always terminates in Step~\ref{step:return_sols} and returns a solution. 
	For any $\sol{S}$ obtained in Step~\ref{step:return_sols}, 
	we have $\|\sol{S}\|_0\le k$ and  
	\begin{align}
	&\P(\sol{S}) \\
	\le&\ \low{S} + \Delta & &\because\text{Termination condition in Step~\ref{step:termination}}\\
	\le&\ \low{S^*} + \Delta & &\because\text{$\low{S}$ is the smallest in \minheap}\\
	\le&\  \F(S^*) + \Delta & &\because\text{$\low{S^*}\le\F(S^*)$ from~\eqref{ineq:admissible}}\\
	=&\ \P(x^*) + \Delta. & &\because\text{$F(S^*)=P(x^*)$ from~\eqref{eq:f_eq_p}}
	\end{align}
	Thus the theorem holds.
\end{proof}
By using $\Delta\ge0$ that is as small as the 
machine epsilon, 
we can obtain an exact BFS. 
As $\Delta$ becomes larger, BFS can terminate earlier. 	
%Empirically, as $\Delta$ becomes larger, BFS becomes faster. 
Therefore, we can use $\Delta$ as a hyper-parameter that 
controls the trade-off between    
the running time and accuracy. 
Similar techniques 
are considered in the field of heuristic search \citep{ebendt2009weighted,valenzano2013using}.

\section{SUBTREE SOLVER}\label{section:oracle}  
We develop \solver\    
that satisfies requirements \eqref{ineq:admissible} and~\eqref{eq:complete}. 
Although the BFS framework does not require us to specify the form of $P(x)$, 
we here assume it can be written as 
follows  for designing \solver: 
\begin{equation}
P(x)=L(Ax)+\frac{\lambda}{2}\|x\|^2,
\label{eq:P}
\end{equation}
where 
$L(\cdot)$ is a convex function,  
$A\in\R^{n\times d}$ is a design matrix, 
$\lambda>0$ is a regularization parameter, 
and 
$\|{\cdot}\|$ denotes the $\ell_2$-norm. 
Since $\ell_2$-regularization is often used to prevent over-fitting 
and $L(\cdot)$ accepts various convex loss functions, 
$P(x)$ of form~\eqref{eq:P} appears in many practical problems (see, e.g.,~\citep{liu2017dual}). 
In \Cref{subsec:example}, we list some examples of loss functions. 
We also assume that the following minimization problem can be 
solved exactly for any $S\in\V$: 
\begin{align}\label{prob:convex} 
\minimize_{x\in\R^d}\ \P(x) 
& & 
\subto\ \supp(x)\subseteq S, 
\end{align}
which can be seen as 
unconstrained minimization of a strongly convex function 
with $|S|$ variables. 
If $P(\cdot)$ is quadratic, 
we can solve it by computing a pseudo-inverse matrix. 
Given more general $P(\cdot)$, 
we can use iterative methods such as~\citep{shalev2016accelerated} 
to solve problem~\eqref{prob:convex}.   

\subsection{Computing $\low{S}$ and $\sol{S}$}\label{subsection:diht} 
Let $S\in\V$ be any node and 
define
$s\coloneqq|S|$, 
$S_\myle \coloneqq \{i \in \dset \mid \ i \myle  \max S \}$, 
and
$S_\myg \coloneqq \{i \in \dset \mid \ i \myg \max S \}$; 
note that $(S_\myle, S_\myg)$ forms a partition of $\dset$. 
For any $x\in\R^d$, 
$x_S\in\R^{s}$ denotes a restricted vector consisting of 
$x_i\in\R$ ($i\in S$). 
Similarly, 
$A_S\in\R^{n\times s}$ denotes a sub-matrix of $A\in\R^{n\times d}$ whose 
column indices are restricted to $S$. 
Given any positive integers $j$, $m$, and $z\in\R^m$, 
we define $\Tl_j(z)\in\R^m$ as follows: 
$\Tl_j(z)$ preserves (up to) $j$ entries of $z$ chosen 
in an non-increasing order of $|z_i|$ and 
sets the rest at $0$. 
We let $\topknorm{\cdot}{j}$ denote the top-$j$ $\ell_2$-norm; 
i.e., $\topknorm{z}{j}\coloneqq\|\Tl_j(z)\|$. 
Given any convex function $f:\R^m\to\R$, 
we denote its convex conjugate by
$f^*(\beta)\coloneqq\sup_{y\in\R^m}\{\innprod{\beta}{y}-f(y)\}$. 
We define $\prox_{f}(x)\coloneqq{\argmin}_{y\in\R^m}\left\{f(y)+\frac{1}{2}\|x-y\|^2 \right\}$.

A high-level sketch of \solver{} is provided in \Cref{alg:oracle}, 
which consists of three parts: 
Steps~\ref{step:s_eq_k}--\ref{step:solver_eq}, 
Steps~\ref{step:s_le_k}--\ref{step:solver_le}, 
and 
Steps~\ref{step:else}--\ref{step:solver_comp}. 
Every part computes $\low{S}$ and $\sol{S}$ that satisfy 
\eqref{ineq:admissible} 
and 
$\|\sol{S}\|_0\le k$, 
and the first part is needed to satisfy~\eqref{eq:complete}. 
While the first two parts consider some easy cases, 
the last part deals with the most important case and 
requires careful discussion. 
Below we explain each part separately.

\begin{algorithm}[t]
	{\fontsize{9pt}{10pt}\selectfont
		\caption{$\solver(S)$}\label{alg:oracle}
		\begin{algorithmic}[1]
			\If{$|S|=k$} \label{step:s_eq_k}
			\State $\sol{S}\gets\argmin_{\supp(x)\subseteq S} \P(x)$ 
			and 
			$\low{S}\gets\P(\sol{S})$  \label{step:solver_eq}
			\ElsIf{$|S|+|S_\myg|\le k$} \label{step:s_le_k}
			\State $\sol{S}\gets \!\!\argmin_{\supp(x)\subseteq S\cup S_\myg}\!\!\P(x)$
			and
			$\low{S}\gets\P(\sol{S})$  \label{step:solver_le}
			\Else \label{step:else}
			\State Compute $\low{S}$ and $\sol{S}$ with \Cref{alg:pdal}. \label{step:solver_comp}
			\EndIf
			\State \Return $\low{S}$, $\sol{S}$
		\end{algorithmic}
	}
\end{algorithm}

\paragraph{Steps~\ref{step:s_eq_k}--\ref{step:solver_eq}.} 
If $|S|=k$, 
$F(S)=\min_{\supp(x)\subseteq S}\P(x)$ holds. 
Therefore, 
$\sol{S}$ and $\low{S}$ 
obtained in 
Step~\ref{step:solver_eq} satisfy 
$\|\sol{S}\|_0\le k$ 
and 
$\F(S) = \P(\sol{S}) = \low{S} \le \low{S} + \Delta$. 
Since $\V^*$ always includes some $S$ of size $k$, %such that $|S|=k$, 
we can guarantee that \solver\ satisfies 
\eqref{eq:complete}. 
%without violating~\eqref{ineq:admissible}. 

\paragraph{Steps~\ref{step:s_le_k}--\ref{step:solver_le}.}
If $|S|+|S_\myg|\le k$, 
then  
$\desc(S) = \{S\cup S^\prime  \relmid{} S^\prime \subseteq S_\myg  \}$, 
and thus $F(S)=\min_{\supp(x)\subseteq S\cup S_\myg}\P(x)$ holds. 
Therefore, 
$\sol{S}$ and $\low{S}$ 
obtained in 
Step~\ref{step:solver_le} satisfy 
$\|\sol{S}\|_0\le k$ 
and 
$\F(S) = \P(\sol{S}) = \low{S}$; 
i.e.,~\eqref{ineq:admissible} holds with equality.

\paragraph{Steps~\ref{step:else}--\ref{step:solver_comp}.}
We consider the case 
where neither $|S|=k$ nor $|S|+|S_\myg|\le k$ holds. 
%
%We consider computing the lower bound $\low{S}$ of $F(S)$.  
Note that $F(S)$ is defined as the minimum value of non-convex minimization problem~\eqref{prob:main}, whose 
lower bound cannot be obtained with standard convex relaxation; 
e.g., 
an $\ell_1$-relaxation-like approach does not always give 
lower bounds. 
For the case of $\ell_0$-constraint minimization (i.e., $\|x\|_0\le k$), 
\citet{liu2017dual} has provided a technique for 
deriving a lower bound.  
Unlike their case, 
the constraint in~\eqref{prob:main} is given by $x\in U(S)$, 
but we can leverage their idea to obtain a lower bound of $F(S)$. 
From the Fenchel--Young inequality, 
$L(Ax) + L^*(\beta) \geq \innprod{Ax}{\beta}$, we obtain
\begin{align}
\begin{aligned}
F(S) 
&=
\min_{x\in \U(S)}\!
\Big\{
L(Ax) + \frac{\lambda}{2}\|x\|^2 
\Big\}\\
&\ge
\min_{x \in \U(S)}\!
\Big\{
\innprod{Ax}{\beta}
-L^*(\beta) + \frac{\lambda}{2} \|x\|^2
\Big\}\\
&=
- L^*(\beta) - \frac{1}{2\lambda} \|A_S^\top \beta\|^2
- \frac{1}{2\lambda}\|A_{S_>}^\top \beta\|^2_{k-s, 2}\\
&\eqqcolon D(\beta; S)
\end{aligned}
\label{ineq:weakdual}
\end{align}
for any $\beta \in \R^n$.  
Therefore, 
once $\beta$ is fixed, we can use $\low{S} = D(\beta; S)$ 
as a lower bound of $F(S)$. 
\begin{algorithm}[tb]
	{\fontsize{9pt}{10pt}\selectfont
		\caption{Computation of $\low{S}$ and $\sol{S}$}
		\label{alg:pdal}
		\begin{algorithmic}[1]
			\State Initialize $\beta^{0}$, $y^{0}$, $\tau_0$, and $\rho_0$. 
			\Comment{Warm-start} \label{step:initialize}
			\State Let $\theta_0\gets 1$ and 
			fix $\gamma>0$. 
			\label{step:pdal_init}
			\State $\dmax\gets D(\beta^{0};S)$     
			\For{$t=1,2,\dots$}
			\State{$\beta^{t} \gets \prox_{\tau_{t-1}L^*}(\beta^{t-1} - \tau_{t-1}Ay^{t-1})$} \label{step:pdal_beta_update}
			\State $\dmax\gets\max\{\dmax, D(\beta^{t};S)\}$ \label{step:pdal_dmax}			
			\State{$\rho_t \gets \rho_{t-1}(1 + \gamma \tau_{t-1})$}
			\State{$\tau_{t} \gets \tau_{t-1} \sqrt{\frac{\rho_{t-1}}{\rho_t}(1 + \theta_{t-1})}$}
			\Loop \Comment{Linesearch loop}
			\State{$\theta_t \gets \frac{\tau_{t}}{\tau_{t-1}}$}
			\State
			$\bar y^{t} \gets y^{t-1} + \rho_t\tau_{t} A^\top (\beta^{t} + \theta_{t}(\beta^{t} - \beta^{t-1}))$
			\State
			\begin{tabular}{@{}l}
				$y^{t}_S \gets \frac{1}{1 + \lambda \rho_t\tau_{t}} \bar y^{t}_S$, 
				\quad 
				$y^{t}_{S_\leq \setminus S} \gets 0$, 
				\quad 
				and \\
				$y^{t}_{S_>} \gets \prox_{\rho_t\tau_{t} (\frac{1}{2\lambda} \|\cdot\|^2_{k-s,2})^*}(\bar y^{t}_{S_>})$
			\end{tabular} \label{step:prox_topk2}
			\If {
				%\resizebox{.70\linewidth}{!}{
				$\sqrt{\rho_t}\tau_t\|A^\top(y^{t} - y^{t-1})\| \leq \|y^{t} - y^{t-1}\|$
				%}
			}
			\State{\textbf{break}}
			\EndIf
			\State{$\tau_t \gets 0.5 \times \tau_t $}
			\EndLoop
			\If{converged} \label{step:pdal_ifconverged}
			\State $\low{S}\gets \dmax$ 
			\State $\sol{S}\gets\argmin_{\supp(x)\subseteq \supp(\Tl_{k}(y^{t}))} \P(x)$  \label{step:sol}
			\State \Return $\low{S}$,  $\sol{S}$ 
			\EndIf
			\EndFor
		\end{algorithmic}
	}
\end{algorithm}
In practice, 
BFS becomes faster as the lower bound becomes larger. 
Thus we consider obtaining a large $D(\beta; S)$ value by 
(approximately) solving the following non-smooth concave maximization problem: 
%seek to find $\beta$ that yields a large $D(\beta; S)$ value. 
\begin{align}
\maximize_{\beta \in \R^n} \quad D(\beta; S).
%=
%- L^*(\beta) - \frac{1}{2\lambda} \|A_S^\top \beta\|^2
%- \frac{1}{2\lambda}\|A_{S_>}^\top \beta\|^2_{k-s, 2}
\label{eq:dualProblem}
\end{align}
\Cref{alg:pdal} presents a maximization method for 
problem~\eqref{eq:dualProblem}, 
which is based on 
%the latest non-smooth convex optimization algorithm, called 
the primal-dual algorithm with linesearch (PDAL)~\citep{malitsky2018first}. 
We may also use the supergradient ascent as a simple alternative to PDAL; 
we here employ PDAL to enhance scalability of BFS. 
(see, \Cref{app:sga_pdal} for details).
If better methods for problem~\eqref{eq:dualProblem} are available, 
we can use them. 
Below we detail \Cref{alg:pdal}. 
In Step~\ref{step:initialize}, we initialize the parameters with the warm-start method detailed in \Cref{subsec:acceleration}. 
In Step~\ref{step:pdal_init}, 
we let $\gamma>0$ be sufficiently small so that 
the $1/\gamma$-smoothness of $L(\cdot)$ holds, 
while larger $\gamma$ makes PDAL faster.  
%	\rr{If $L(\cdot)$ is Huber or quadratic losses, %used in~\Cref{section:experiments}, 	
%	we can set $\gamma \gets n$.  %and we do so in our experiments.
%	}   
How to 
compute $\prox(\cdot)$ (Steps~\ref{step:pdal_beta_update} and~\ref{step:prox_topk2}) 
is detailed in \Cref{app:prox}. 
In Step~\ref{step:sol}, we compute a feasible solution $\sol{S}$ from 
the primal solution, $y^t$.
We now explain how to detect convergence in Step~\ref{step:pdal_ifconverged}, 
which requires us to consider the following two issues: 
(I) 
It would be ideal if we could use the relative error, 
$(D(\beta^{t};S) - D(\beta^{t-1};S))/\max_{\beta \in \R^n} D(\beta;S)$, 
for detecting convergence, 
but the denominator is unavailable. 
(II) 
$D(\beta^{t};S)$ does not always increase with $t$, 
while we want to make the output, $\low{S}=\dmax$, 
as large as possible. 
We first address (I). 
Let $\pmin$ be the best current objective value 
when $\solver(S)$ is invoked, which 
we maintain as in \Cref{alg:bfs}. 
If \Cref{alg:pdal} is not force-quitted by the pruning procedure, 
we always have $D(\beta;S)<\pmin$ as detailed in \Cref{subsec:acceleration}.  
Furthermore,  
\Cref{alg:pdal} aims to maximize $D(\beta;S)$.   
These facts suggest that $\pmin$ would be a good surrogate of $\max_{\beta \in \R^n} D(\beta;S)$. 
Hence we use 
$(D(\beta^{t};S) - D(\beta^{t-1};S))/\pmin\le\epsilon$ 
as a termination condition, where $\epsilon>0$ is a small constant 
that controls the accuracy of \solver.  
This  condition alone is, however, insufficient due to issue (II); 
i.e., 
$D(\beta^{t};S) - D(\beta^{t-1};S)$ can be negative even though 
$\low{S}=\dmax$ is small. 
To resolve this problem, 
we employ an additional termination condition, $D(\beta^{t};S)\ge \dmax$, 
which prevents \Cref{alg:pdal} from outputting small $\low{S}$. 
If both conditions are satisfied, 
we regard the for loop as having converged.

\subsection{Acceleration Techniques}\label{subsec:acceleration}
We present two acceleration techniques: 
a warm-start method and pruning via force-quit, 
which is mentioned in \Cref{subsection:bfs}. 
In \Cref{app:ablation}, 
ablation experiments confirm that the combination of the two acceleration techniques 
greatly speeds up BFS. 

\newcommand{\Sparent}{{S^\prime}}
\paragraph{Warm-start by Inheritance.}
We detail how to initialize 
$\beta^{0}$, $y^{0}$, $\tau_0$, and $\rho_0$ with the warm-start method. 
When executing \Cref{alg:pdal} with $S=\emptyset$ at the beginning of BFS, 
we set 
$\beta^{0}\gets0$, 
$y^{0}\gets0$, 
$\tau_0\gets1/\|A\|_2$, and 
$\rho_0\gets1$, 
where $\|A\|_2$ is the largest singular value of $A$. 
We now suppose that $\Sparent\in\V$ is popped from $\minheap$ 
and that we are about to compute $\low{S}$ and $\sol{S}$ 
with \Cref{alg:pdal}, where $(\Sparent,S)\in E$. 
Since $S$ is obtained by adding only one element, $\max S$, 
to $\Sparent$, 
$D(\beta;\Sparent)$ and $D(\beta;S)$ 
are expected to have 
similar maximizers. 
Taking this into account, we set 
$\beta^{0}$, $y^{0}$, $\tau_0$, and $\rho_0$ 
at those obtained in the last iteration of \Cref{alg:pdal} 
invoked  by \solver($\Sparent$). 
Namely, 
\Cref{alg:pdal} inherits 
$\beta^{0}$, $y^{0}$, $\tau_0$, and $\rho_0$ 
from the parent node to become warm-started.  
We can easily confirm that 
$D(\beta^{0}; \Sparent)\le D(\beta^{0};S)$ 
holds for any $\beta^{0}\in\R^n$.

\paragraph{Pruning via Force-quit.} 
As mentioned in \Cref{subsection:bfs}, 
force-quitting $\solver(T)$ can
accelerate BFS, 
but it involves detecting $\low{T}>\pmin$;     
we explain how to do this. 
Since problem~\eqref{prob:convex} is assumed to be solved efficiently, 
detecting $\low{T}>\pmin$ 
for the cases of 
Steps~\ref{step:s_eq_k}--\ref{step:solver_eq} 
and  
Steps~\ref{step:s_le_k}--\ref{step:solver_le} 
in \Cref{alg:oracle} 
is easy; 
i.e., we check whether $\low{T}=P(\sol{T})>\pmin$ holds or not.  	 
Below we focus on the case of
Steps~\ref{step:else}--\ref{step:solver_comp}.  
While executing \Cref{alg:pdal}, 
once $D(\beta^{t};T)>\P_{\min}$ occurs for some $t$,  
then we have $\low{T}>\P_{\min}$  due to Step~\ref{step:pdal_dmax}. 
In this case, we force-quit \Cref{alg:pdal}, 
and continue BFS without pushing 
$T$ onto \minheap.  

\newcommand{\Llogistic}{L_{\text{logistic}}}
\newcommand{\Lquadratic}{L_{\text{quadratic}}}
\subsection{Examples of Loss Functions}\label{subsec:example}
We detail three examples of convex loss functions $L(\cdot)$: 
quadratic, Huber, and logistic loss functions. 
We will use them in the experiments. 	
All of the functions are defined with 
design matrix $A=[a_1,\dots,a_n]^\top\in\R^{n\times d}$ 
and 
observation vector $b=[b_1,\dots,b_n]^\top\in\R^n$. 

\paragraph{Quadratic Loss.}
The quadratic loss function is 
a widely used loss function defined as   
$
\Lquadratic(Ax)\coloneqq
\frac{1}{2n}\| b - Ax \|^2 
$.
Note that $\Lquadratic(\cdot)$ is $1/n$-smooth, and so we can set $\gamma=n$ 
in \Cref{alg:pdal}.

\paragraph{Huber Loss.}
When observation vector $b$ contains outliers, 
the Huber loss function is known to be effective. 
Given parameter $\delta\ge0$, 
the function  is defined as 
$
\Lhuber(Ax) \coloneqq\frac{1}{n}\sum_{i=1}^n l(a_i^\top x-b_i)
$, 
where 
$l(r)$ is $r^2/2$ if $|r|\le\delta$ 
and $\delta\big(|r|-\delta/2\big)$ otherwise. 
We can confirm that $\Lhuber(\cdot)$ is also $1/n$-smooth.

\paragraph{Logistic Loss.}
When each entry of the observation vector is dichotomous, 
i.e., $b_1,\dots,b_n \in \{-1,1\}$, 
the following logistic loss function is often used: 
$
\Llogistic(Ax)
\coloneqq\frac{1}{n}\sum_{i=1}^n 
(1 + \exp( - b_i \cdot a_i^\top x))
$. 
Note  that $\Llogistic(\cdot)$ is $\frac{1}{4n}$-smooth. 
Although
the proximal operator of this function 
required in \Cref{alg:pdal} 
has no closed expression, 
we can efficiently 
compute it by solving a 1D minimization problem with Newton's method as in~\citep[Appendix A]{defazio2016simple}.

\section{EXPERIMENTS}\label{section:experiments}
We evaluate our BFS via experiments. 
In \Cref{subsec:synthetic}, 
we use synthetic instances with Huber and logistic loss functions; 
we thus confirm that our BFS can solve NSM instances to which 
the MIP-based method is inapplicable. 
With the instances, we examine the computation cost of BFS.  
We also demonstrate that our BFS is useful in terms of support recovery; 
this is the first experimental study that examines 
the support recovery performance of exact algorithms for 
NSM instances with non-quadratic loss functions. 
In \Cref{subsec:real}, 
we use two real-world NSM instances with quadratic loss functions, 
and we demonstrate that BFS can run faster than the MIP-based method 
\citep{bertsimas2016best} with the latest commercial solver, Gurobi 8.1.0, 
which we denote simply by MIP in what follows.

For comparison, we employed three inexact methods: 
OMP~\citep{elenberg2018restricted}, 
HTP~\citep{yuan2014gradient}, 
and dual IHT (DIHT)~\citep{liu2017dual}. 
The precision, $\epsilon$, used 
for detecting convergence in \Cref{alg:pdal} 
(see, the last paragraph in \Cref{subsection:diht}), 
as well as those of 
those of HTP and DIHT, were set at $10^{-5}$.
When solving 
problem~\eqref{prob:convex} with non-quadratic objectives, 
we used  
a primal-dual method based on~\citep{shalev2016accelerated};  
with this method we obtained  
solutions whose primal-dual gap was at most $10^{-15}$. 
We regarded numerical errors smaller than $10^{-12}$ as $0$. 

All experiments were conducted 
on a 64-bit Cent6.7 machine 
with Xeon 5E-2687W v3 3.10GHz CPUs and 128 GB of RAM.  
All methods were executed with a single thread. 
BFS, OMP, HTP, and DIHT 
were implemented in Python 3, 
and    
MIP used Gurobi 8.1.0.

\subsection{Synthetic Instances}\label{subsec:synthetic} 
We consider synthetic NSM instances of sparse regression models; 
we estimate $x\in\R^d$, which has a support of size at most $k$, 
from a sample of size $n$. 
Specifically, 
%	given $d$, $k$, and $n$, 
we created the following instances 
with Huber and logistic loss functions, 
which we simply call Huber and logistic instances, respectively, 
in what follows.

\newcommand{\Strue}{{S_\text{true}}}
\paragraph{Huber Instance: Sparse Regression with Noise and Outliers.}
We created $\Strue\subseteq\dset$ 
by randomly sampling $k$ elements from $\dset$, 
which forms the support of the true sparse solution, $\xtrue$.  
We set the $i$th entry of $\xtrue$ at $1$ 
if $i\in\Strue$ and $0$ otherwise. 
We drew each row of 
$A\in\R^{n\times d}$ from 
a lightly correlated $d$-dimensional normal distribution, 
whose 
mean and correlation coefficient were set at $0$ and $0.2$, respectively.   
We normalized each column of $A$ 
so that its $\ell_2$-norm became $1$.   
We then drew each entry of $\xnoise\in\R^d$ 
from the standard normal distribution, 
denoted by $\randn$, 
and rescaled it so that the signal-noise ratio, 
$\|\xtrue\|/\|\xnoise\|$, became $10$.  
We let $b^*=A(\xtrue+\xnoise)$. 
Analogously, we drew each entry of $\bnoise\in\R^n$ 
from $\randn$ and rescaled it so that $\|b^*\|/\|\bnoise\|=10$ held. 
We then randomly chose $\lfloor 0.1n \rfloor$ 
entries from $\bnoise$ and multiplied them by 10; 
we let $b=b^*+\bnoise$. 
Namely, about $10\%$ entries of $b$ are outliers. 
%	We thus obtained $A$ and $b$.  
We used the regularized Huber loss function, 
$\P(x)=\Lhuber(Ax)+\frac{\lambda}{2}\|x\|^2$, 
as an objective function, 
where we let $\delta=1$ and $\lambda=0.001$. 

\paragraph{Logistic Instance: Ill-conditioned Sparse Regression with Dichotomous Observation.}
As with the above setting, we created $\Strue$ and set 
the $i$th entry of $\xtrue$ at $10$ if $i\in\Strue$ and $0$ otherwise. 
We employed ill-conditioned design matrix $A\in\R^{n\times d}$ 
as in \citep{jain2014iterative}: 
We obtained $\hat S\subseteq\dset$ of size $k$ 
by randomly choosing $\ceil{k/2}$ elements from $\Strue$ 
and $\floor{k/2}$ elements from $\dset\backslash\Strue$. 
We then drew each row of $A_{\hat{S}}$ from 
a heavily correlated $k$-dimensional normal distribution 
with mean $0$ and correlation coefficient $0.5$,  
and each row of $A_{\dset\backslash\hat{S}}$ was 
drawn from the above lightly correlated normal distribution of 
dimension $d-k$. 
We then normalized each column of $A$. 
We drew each entry of $b\in\{-1,1\}^n$ from a Bernoulli 
distribution such that $b_i=1$ with a probability of $1/(1+\exp(-a_i^\top \xtrue))$; 
i.e., $\Strue$ represents features that affect the dichotomous outcomes.  
We used the regularized logistic loss function 
$\P(x)=\Llogistic(Ax)+\frac{\lambda}{2}\|x\|^2$ 
as an objective function, 
where we let $\lambda=0.0002$.

\begin{figure}[tb]
	\centering
	%\begin{tabular}[t]{p{0.23\textwidth}p{0.23\textwidth}}
	\begin{minipage}[t]{0.235\textwidth}
		\includegraphics[width=1.0\linewidth]{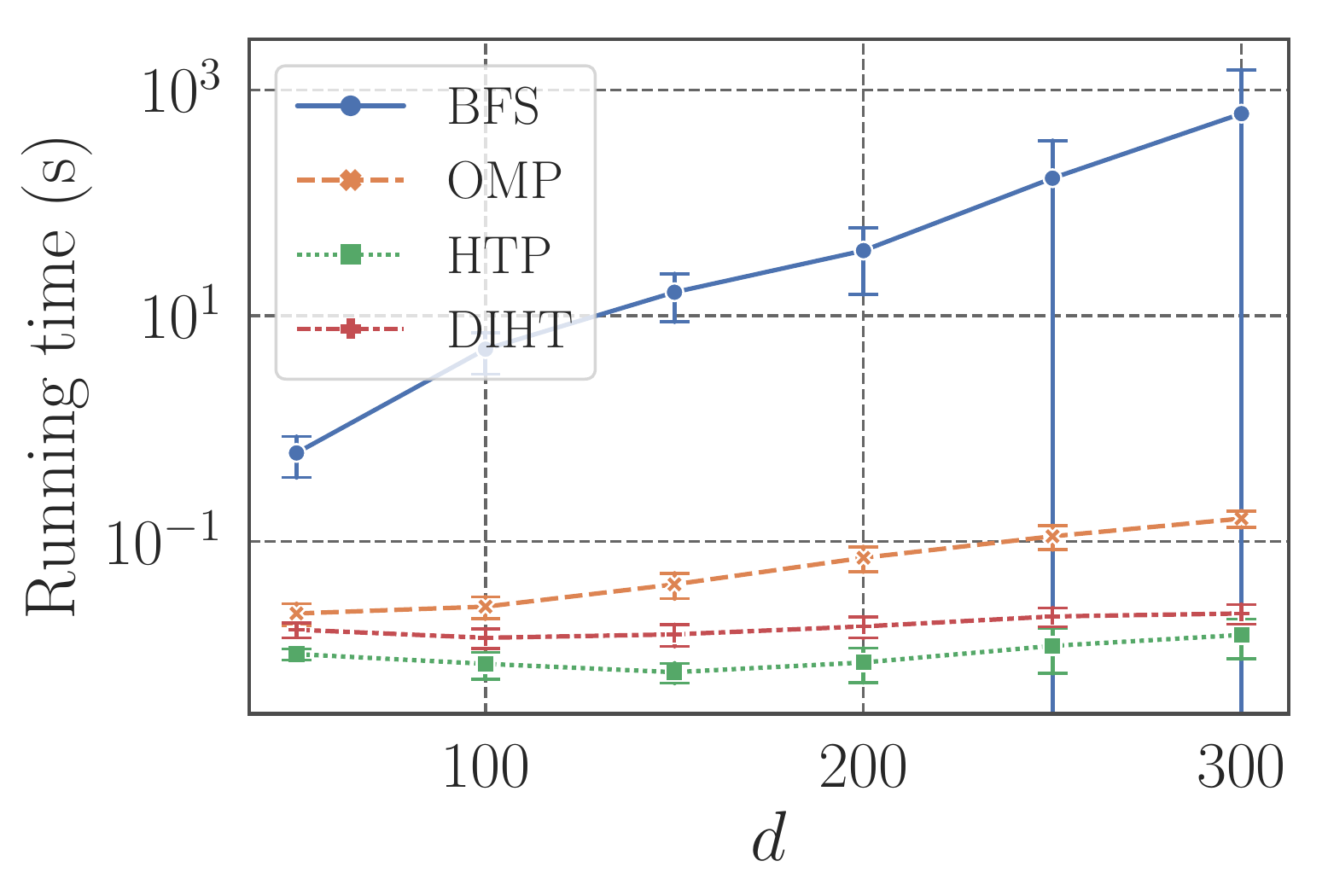}
		\subcaption{Huber, Running Time}
		\label{fig:huber_time}
	\end{minipage}
	\begin{minipage}[t]{0.235\textwidth}
		\includegraphics[width=1.0\linewidth]{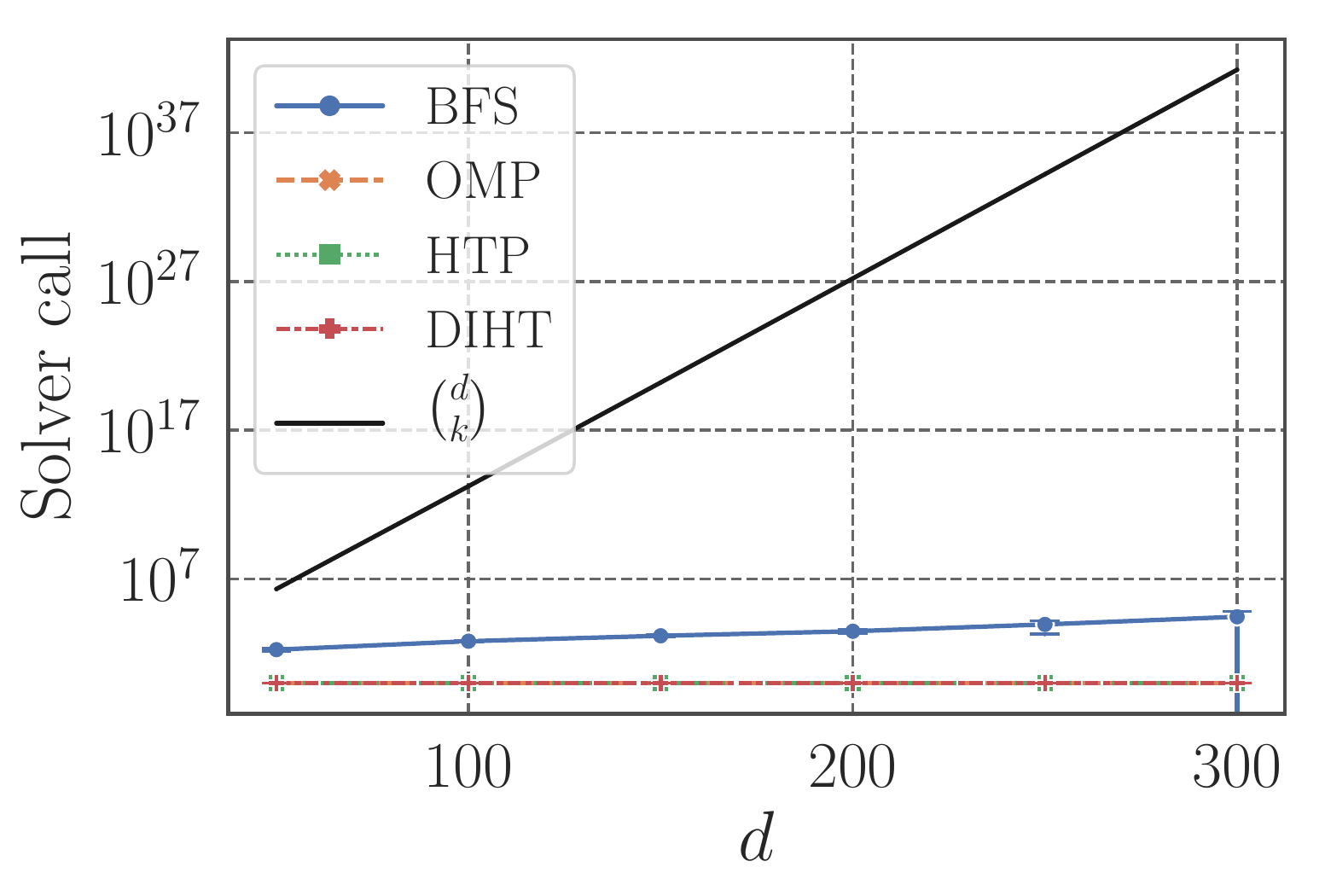}
		\subcaption{Huber, Solver Call}
		\label{fig:huber_call}
	\end{minipage}
	\begin{minipage}[t]{0.235\textwidth}
		\includegraphics[width=1.0\linewidth]{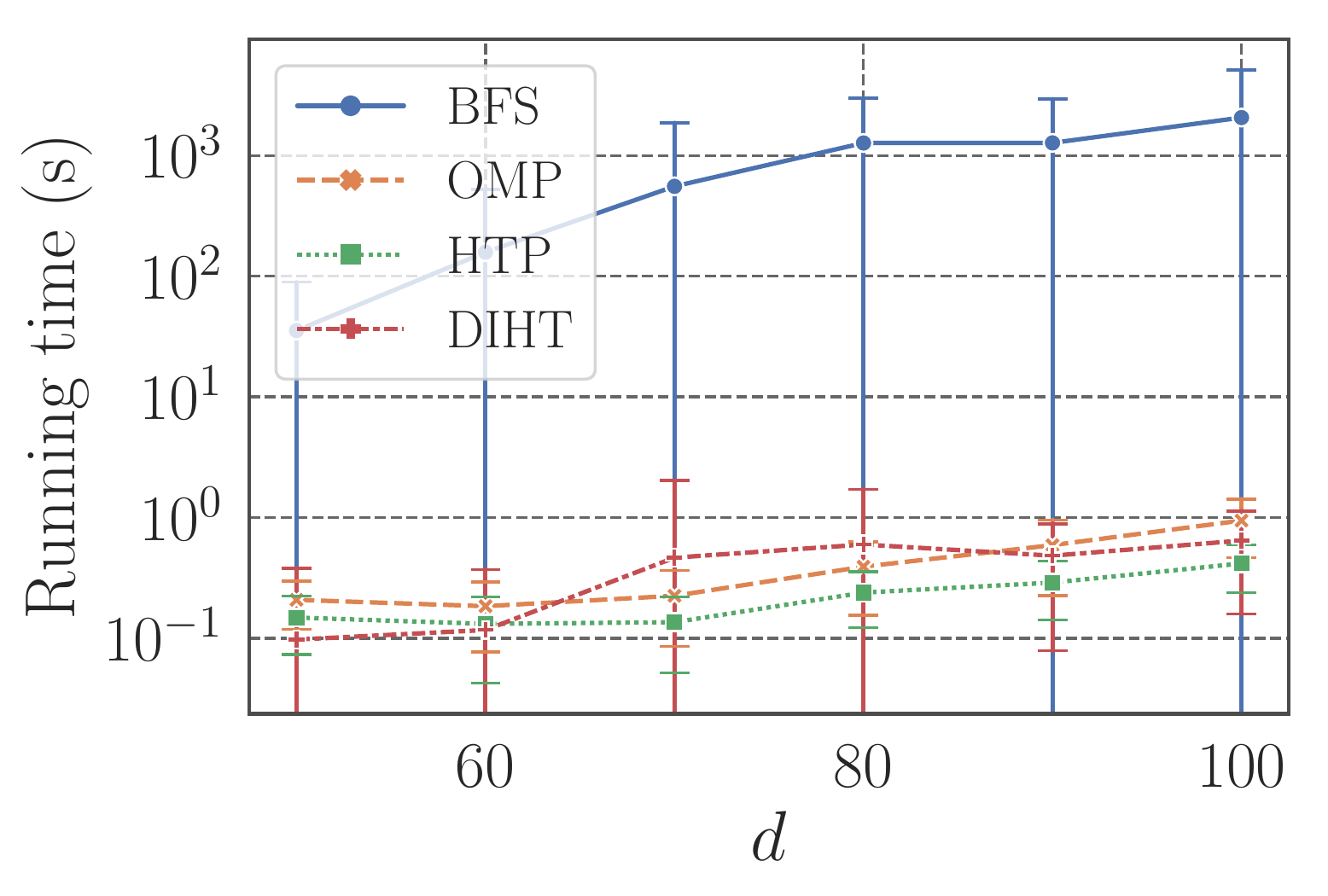}
		\subcaption{Logistic, Running Time}
		\label{fig:logistic_time}
	\end{minipage}
	\begin{minipage}[t]{0.235\textwidth}
		\includegraphics[width=1.0\linewidth]{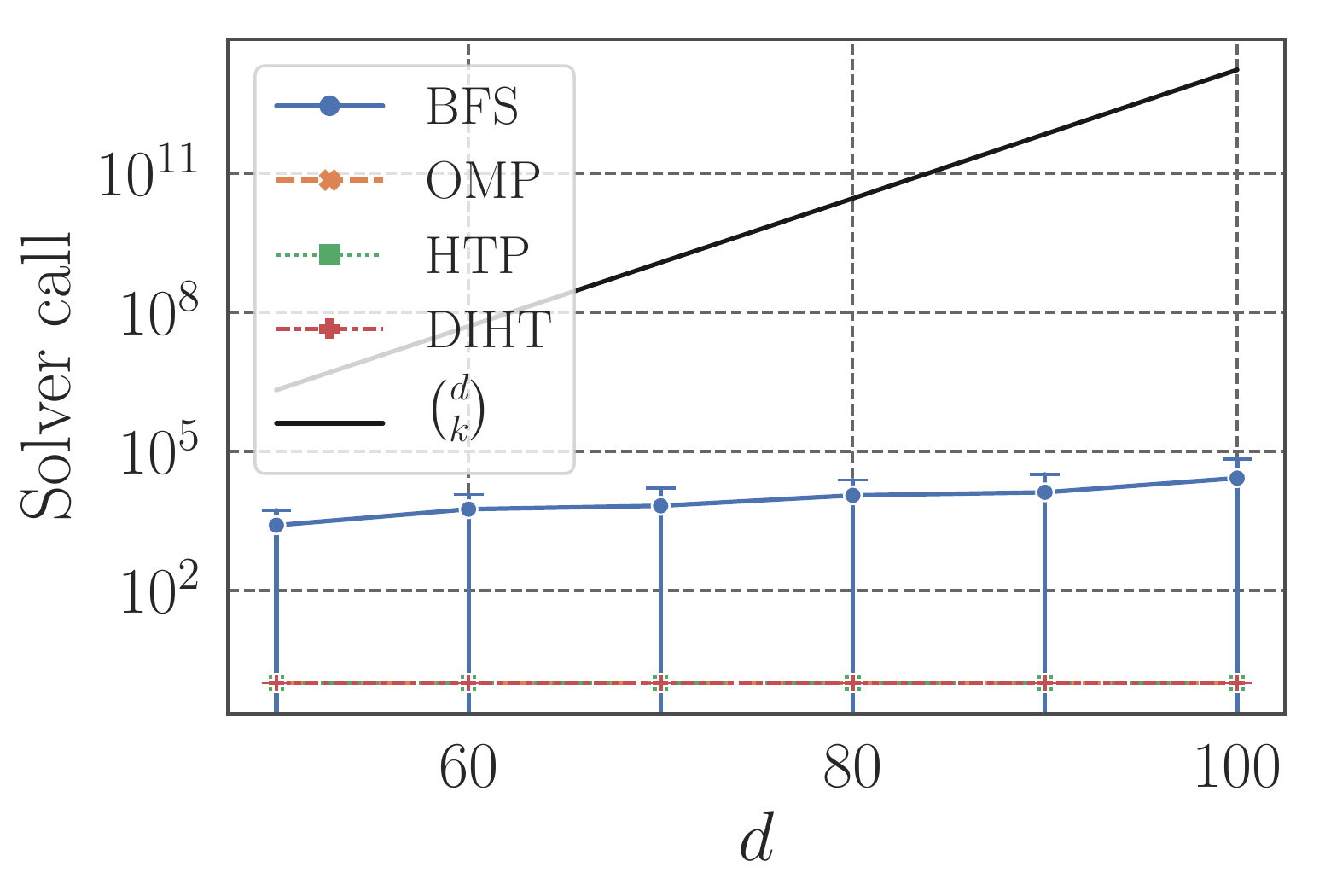}
		\subcaption{Logistic, Solver Call}
		\label{fig:logistic_call}
	\end{minipage}
	%\end{tabular}
	\caption{
		Running Times and Solver Calls. 
		$\binom{d}{k}$ corresponds to the solver call of 
		exhaustive search.
	} 
	\label{fig:syn_plot}
\end{figure}

\subsubsection{Computation Cost}\label{subsubsec:cost}
We created $100$ random Huber and logistic instances 
with $d=50,100,\dots,300$ and $d=50,60,\dots,100$, 
respectively.  
We let $k=0.1k$ and $n=\lfloor 10k\log d \rfloor$. 
\Cref{fig:syn_plot} shows 
the running time  
and solver call, 
which indicates the number of times \solver\ is executed, of each method. 
The solver calls of the inexact methods are regarded as $1$.
Each curve and error bar indicate the mean and standard deviation 
calculated over $100$ instances.  	
For comparison, 
we present the $\binom{d}{k}$ values, 
which correspond to the solver calls of a naive exhaustive search 
that solves problem~\eqref{prob:convex} $\binom{d}{k}$ times. 
We see that BFS is far more efficient than the exhaustive search, 
which is too expensive to be used in practice. 
Note that the size of the state-space tree, $|V|$, 
is at least $\binom{d}{k}$; 
hence the results of solver calls confirm that 
BFS examines only a very small fraction of the tree. 
Furthermore, 
although BFS is slower than the inexact methods on average, 
BFS can sometimes run very fast as indicated by the error bars. 
We also counted the number of solved instances for each method: 
While BFS solved all the $600$ Huber instances and $600$ logistic instances, 
OMP, HTP, and DIHT 
solved 
598, 172, and 233 Huber instance, respectively, 
and 
464, 7, and 82 logistic instances, respectively. 
Note that these results regarding the inexact methods are obtained 
thanks to BFS, 
which always provides optimal solutions 
and enables us to see whether solutions obtained with 
inexact methods are optimal or not.

\begin{figure}[tb]
	\centering
	%\begin{tabular}[t]{p{0.23\textwidth}p{0.23\textwidth}}
	\begin{minipage}[t]{0.235\textwidth}
		\includegraphics[width=1.0\linewidth]{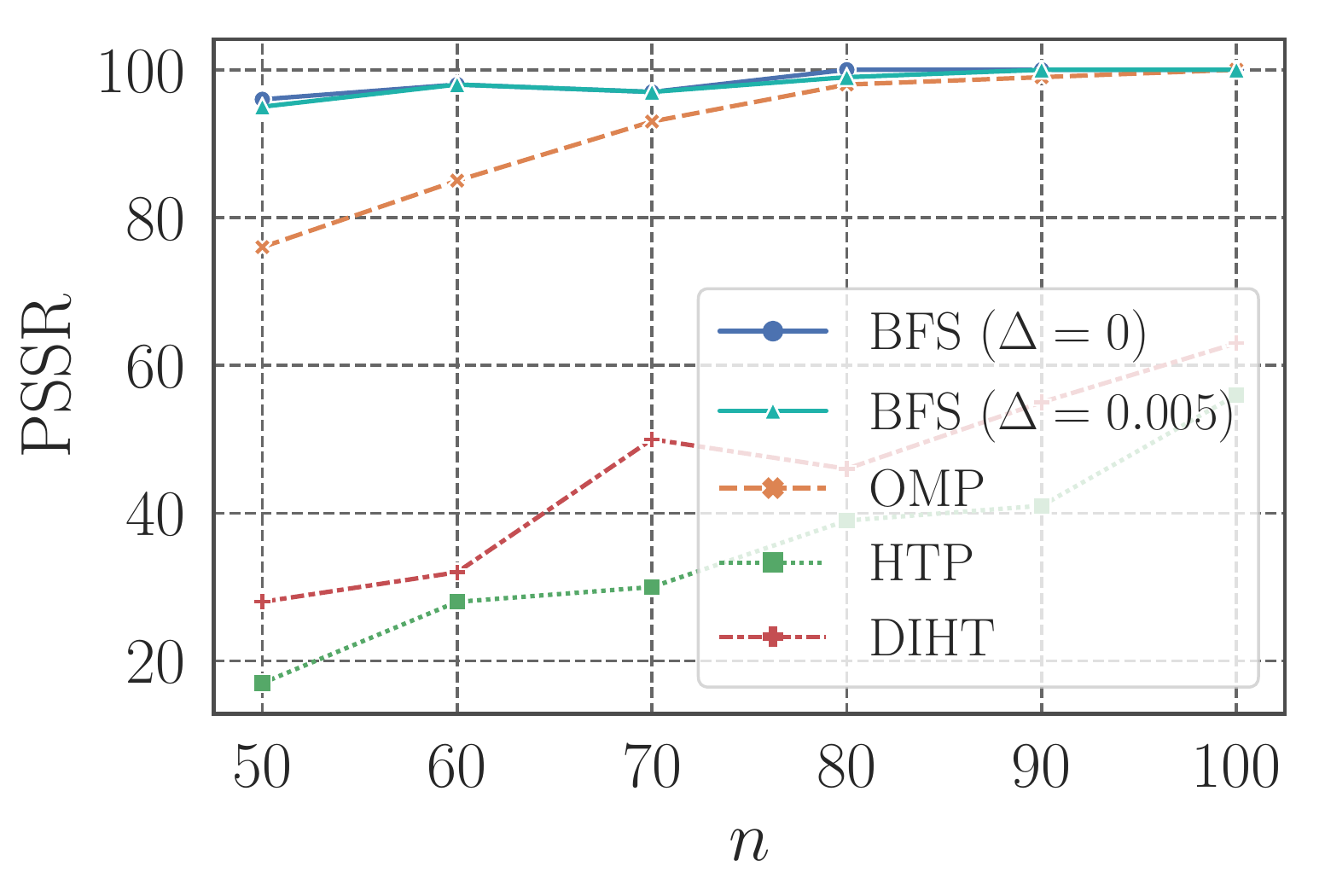}
		\subcaption{Huber, PSSR}
		\label{fig:huber_pssr_w}
	\end{minipage}
	\begin{minipage}[t]{0.235\textwidth}
		\includegraphics[width=1.0\linewidth]{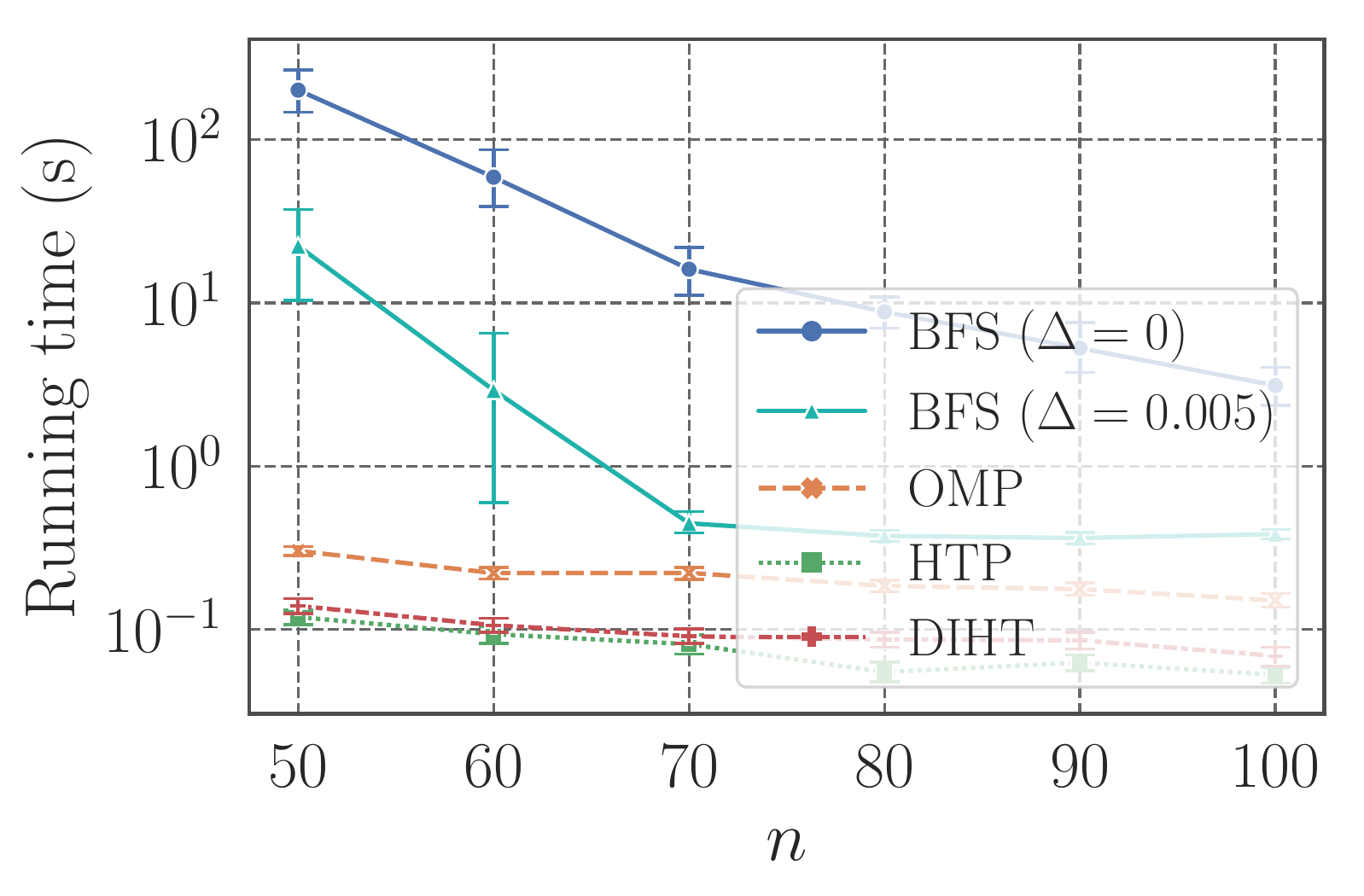}
		\subcaption{Huber, Running Time}
		\label{fig:huber_time_w}
	\end{minipage}
	%		\begin{minipage}[t]{0.235\textwidth}
	%			\includegraphics[width=1.0\linewidth]{./fig/huber_pssr_001.pdf}
	%			\subcaption{Huber, $\lambda=0.01$}
	%			\label{fig:huber_pssr_s}
	%		\end{minipage}
	%		\begin{minipage}[t]{0.235\textwidth}
	%			\includegraphics[width=1.0\linewidth]{./fig/huber_time_001.pdf}
	%			\subcaption{Huber, $\lambda=0.01$}
	%			\label{fig:huber_time_s}
	%		\end{minipage}
	\begin{minipage}[t]{0.235\textwidth}
		\includegraphics[width=1.0\linewidth]{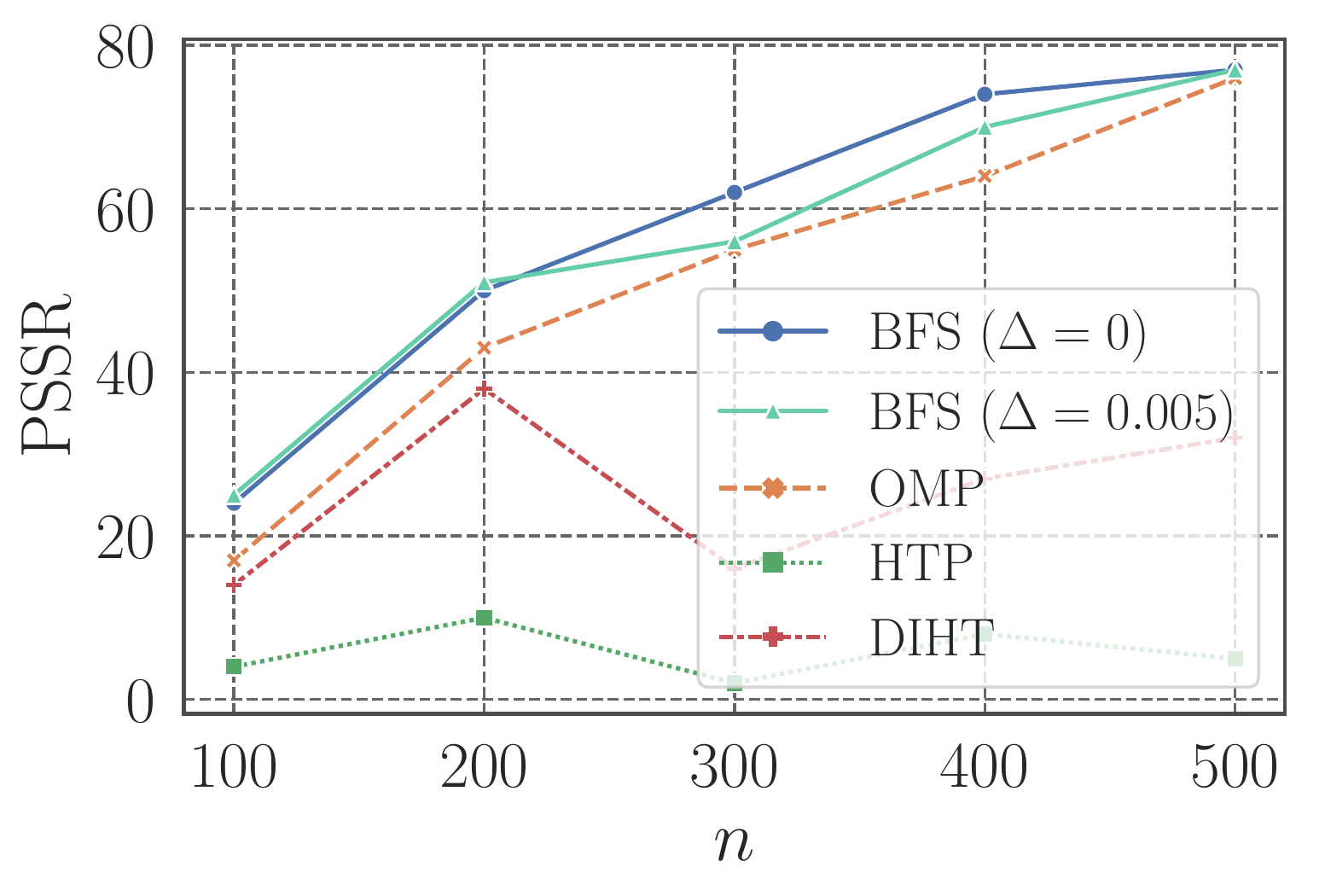}
		\subcaption{Logistic, PSSR}
		\label{fig:logistic_pssr_w}
	\end{minipage}
	\begin{minipage}[t]{0.235\textwidth}
		\includegraphics[width=1.0\linewidth]{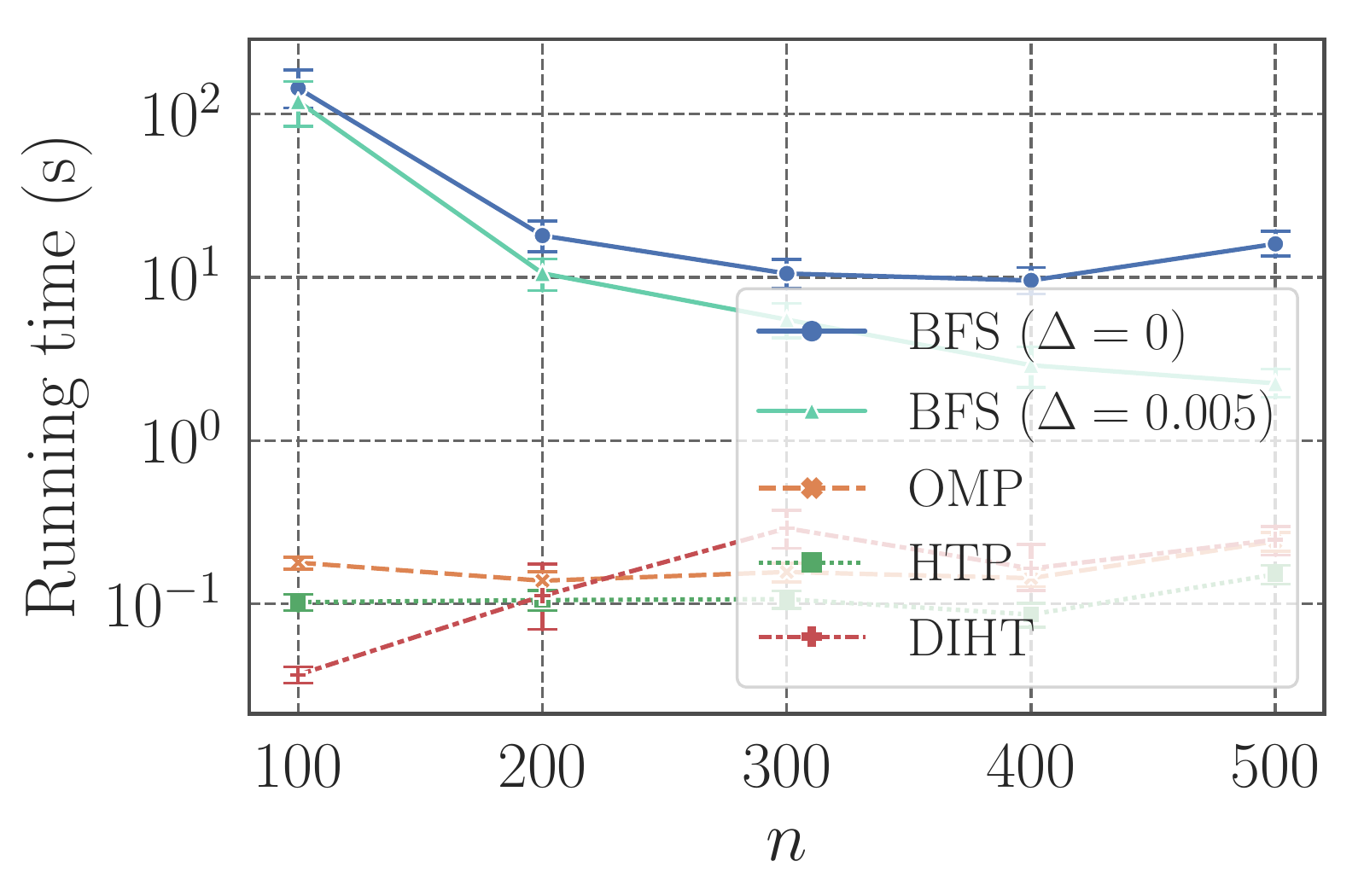}
		\subcaption{Logistic, Running Time}
		\label{fig:logistic_time_w}
	\end{minipage}
	%		\begin{minipage}[t]{0.235\textwidth}
	%			\includegraphics[width=1.0\linewidth]{./fig/logistic_pssr_0002.pdf}
	%			\subcaption{Logistic, $\lambda=0.002$}
	%			\label{fig:logistic_pssr_s}
	%		\end{minipage}
	%		\begin{minipage}[t]{0.235\textwidth}
	%			\includegraphics[width=1.0\linewidth]{./fig/logistic_time_0002.pdf}
	%			\subcaption{Logistic, $\lambda=0.002$}
	%			\label{fig:logistic_time_s}
	%		\end{minipage}
	%\end{tabular}
	\caption{
		PSSR and Running Times.
	} 
	\label{fig:pssr_plot}
\end{figure}

\begin{figure*}[tb]
	\centering
	\begin{minipage}[t]{0.32\textwidth}
		\includegraphics[width=1.0\textwidth]{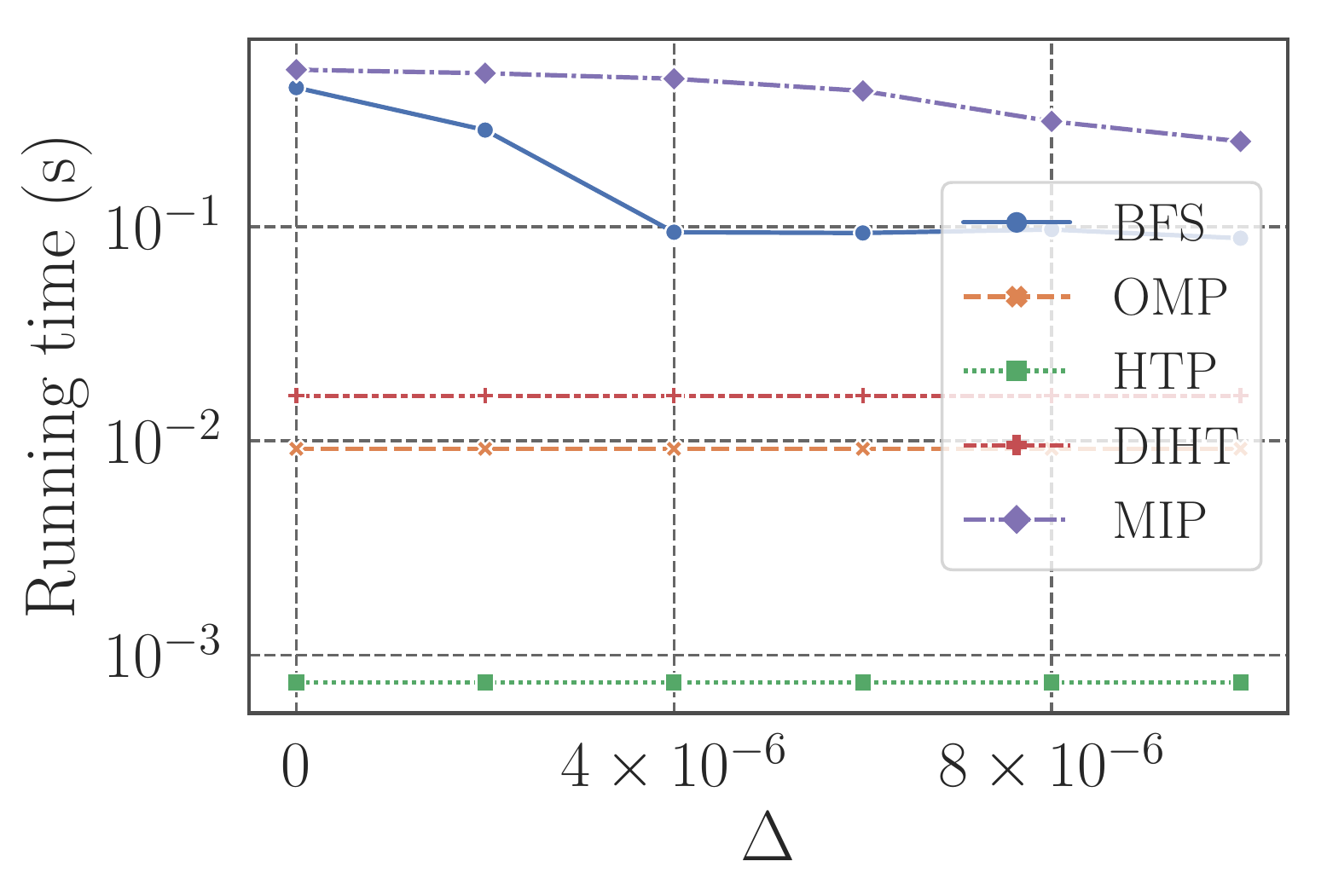}
		\subcaption{Diabetes, Running Time}
		\label{fig:dia_time}
	\end{minipage}
	\begin{minipage}[t]{0.32\textwidth}
		\includegraphics[width=1.0\textwidth]{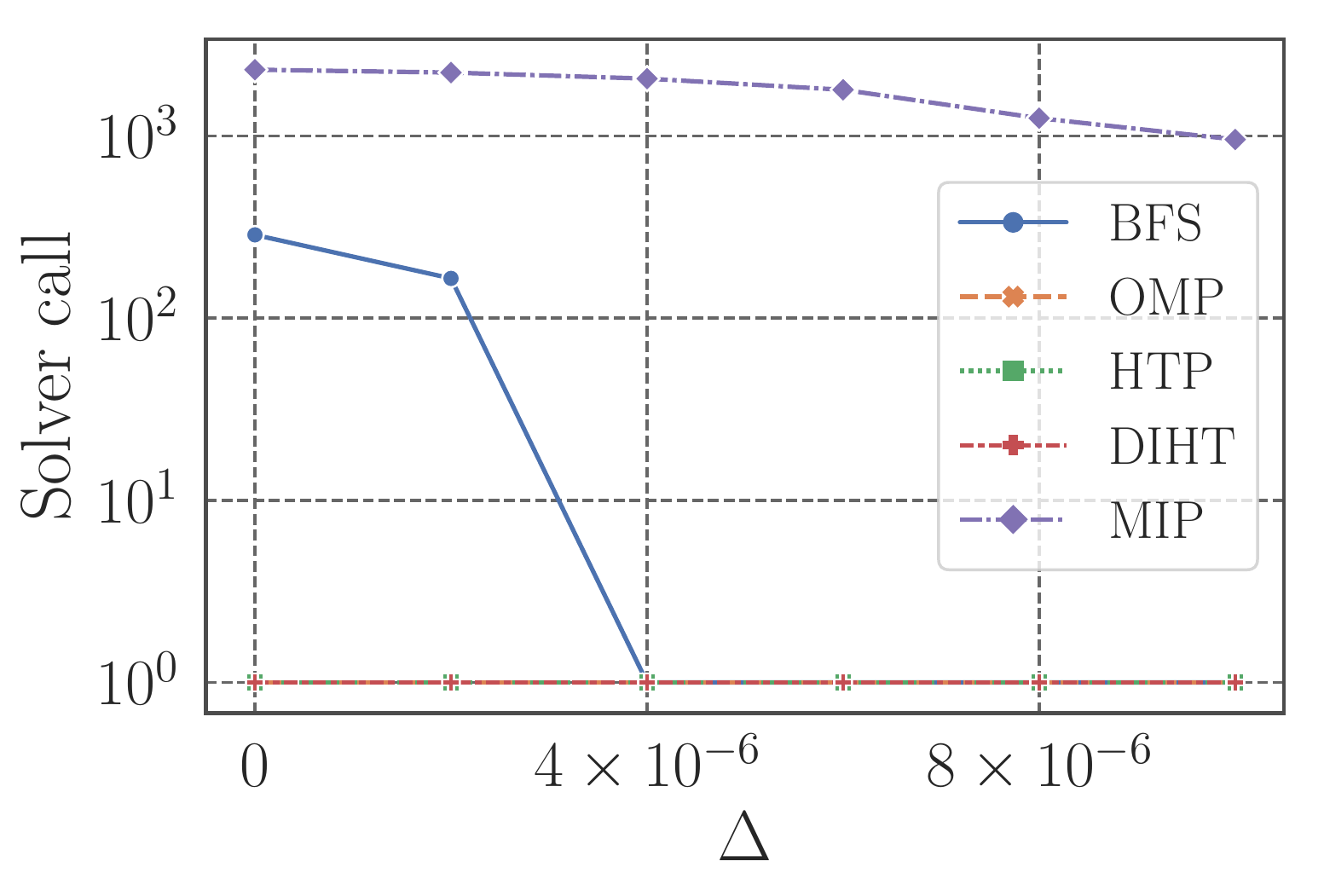}
		\subcaption{Diabetes, Solver Call}
		\label{fig:dia_call}
	\end{minipage}
	\begin{minipage}[t]{0.32\textwidth}
		\includegraphics[width=1.0\textwidth]{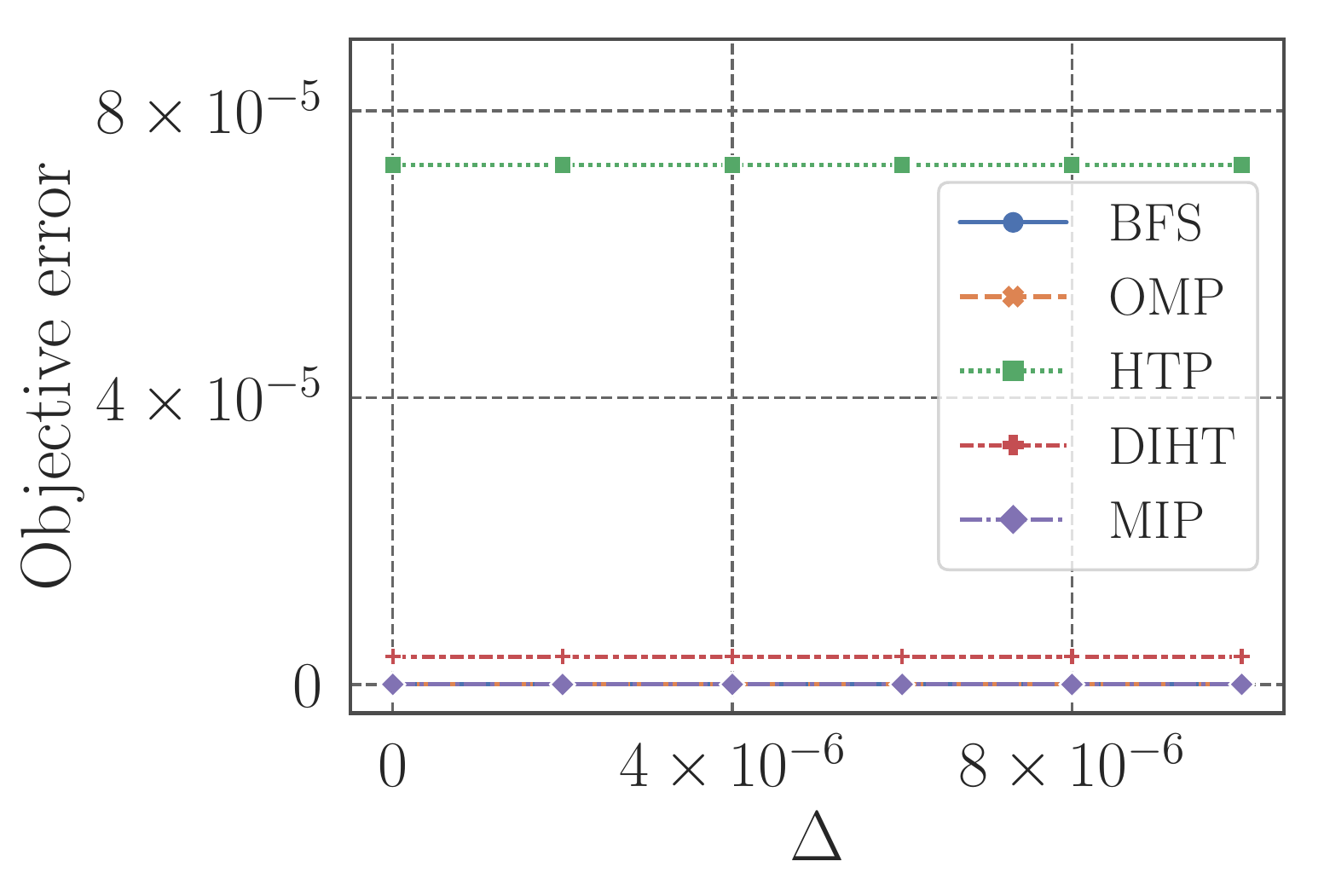}
		\subcaption{Diabetes, Objective Error}
		\label{fig:dia_err}
	\end{minipage}
	\begin{minipage}[t]{0.32\textwidth}
		\includegraphics[width=1.0\textwidth]{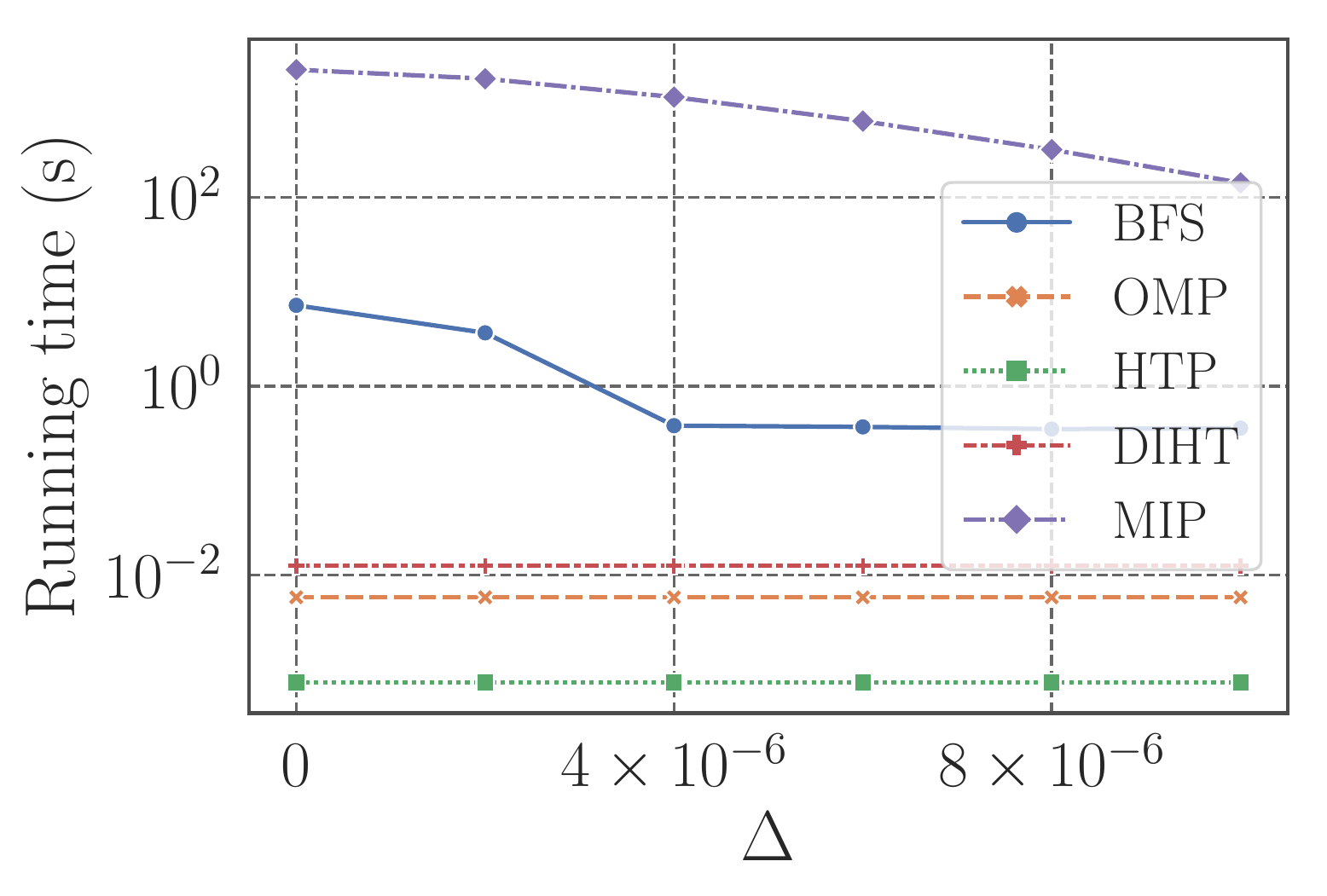}
		\subcaption{Boston, Running Time}
		\label{fig:bos_time}
	\end{minipage}
	\begin{minipage}[t]{0.32\textwidth}
		\includegraphics[width=1.0\textwidth]{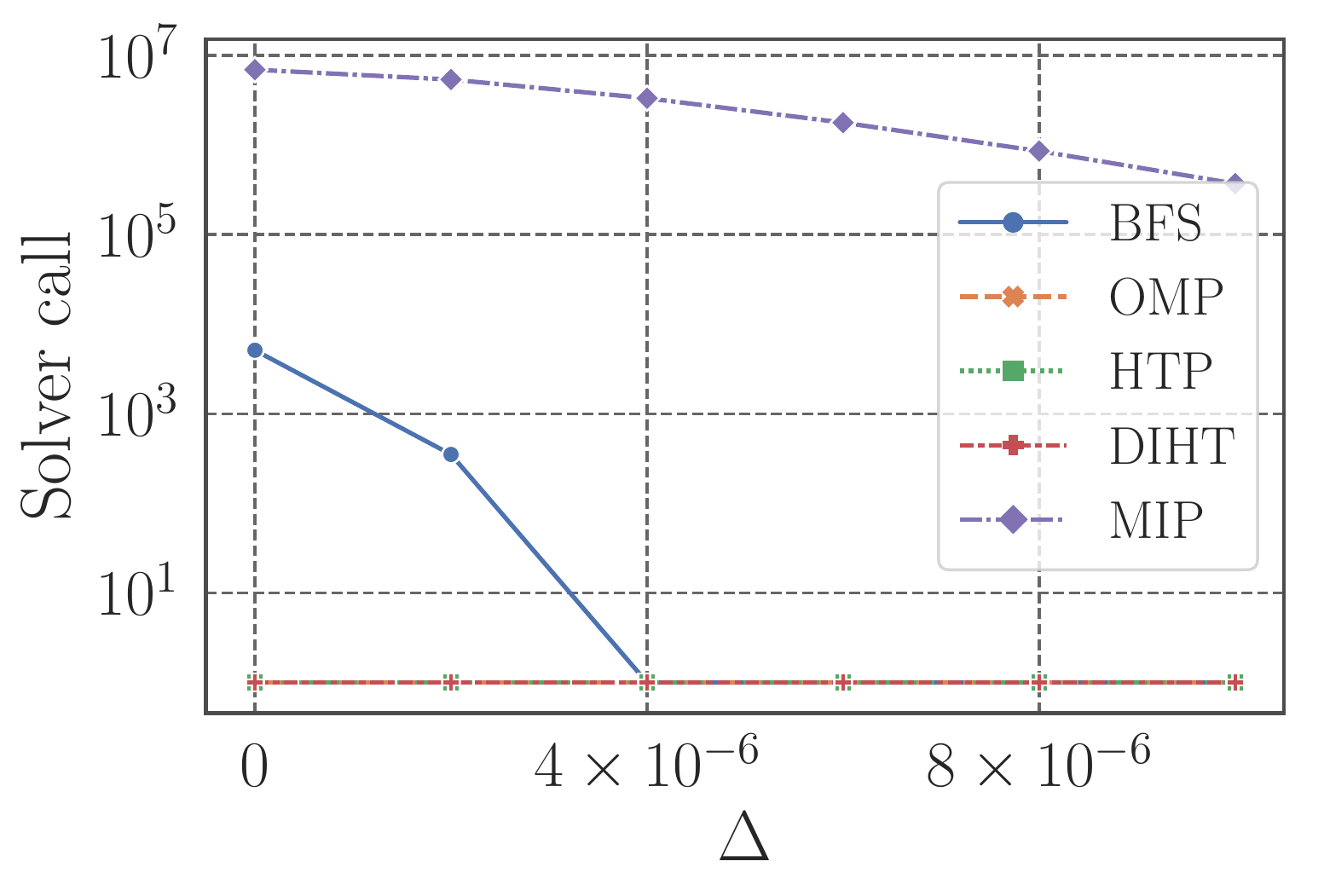}
		\subcaption{Boston, Solver Call}
		\label{fig:bos_call}
	\end{minipage}
	\begin{minipage}[t]{0.32\textwidth}
		\includegraphics[width=1.0\textwidth]{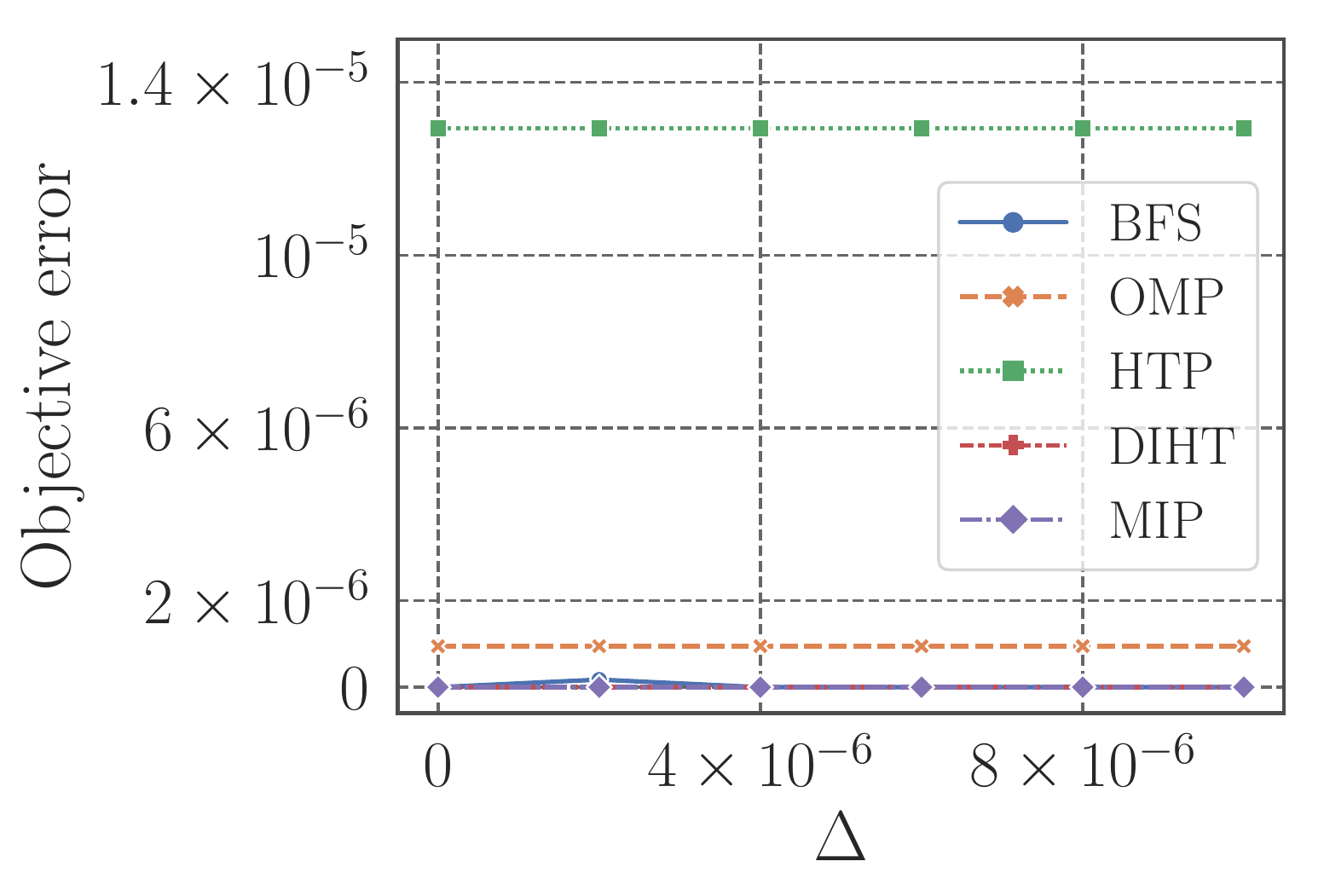}
		\subcaption{Boston, Objective Error}
		\label{fig:bos_err}
	\end{minipage}
	\caption{Running Times, Solver Calls, and Objective Errors. 
		Upper and lower figures show those of 
		the Diabetes and Boston instances, respectively. 
		We provide the results of OMP, HTP, and DIHT for comparison.}
	\label{fig:real}
\end{figure*}

\subsubsection{Support Recovery}\label{subsubsec:pssr}
We used Huber and logistic instances with 
$(d,k)=(50,5)$. 
We randomly generated $100$ Huber instances 
with $n=50,60,\dots,100$ 
and $100$ logistic instances with $n=100,200,\dots,500$.  
We applied the algorithms used in the above section and 
BFS that accepts $\Delta=0.005$ objective errors 
to the instances. 
We evaluated them with the percentage of successful support recovery 
(PSSR) as in \citep{liu2017dual}, 
which counts the number of instances such that output solution $x$ 
satisfies $\supp(x)=\supp(x^*)$ among the $100$ instances. 
We also measured running times of the algorithms. 
\Cref{fig:pssr_plot} summarizes the results. 
BFS achieved higher PSSR performances than the inexact methods 
with both Huber and logistic instances. 
In particular, 
the performance gaps in Huber instances with small $n$ are significant 
as in \Cref{fig:huber_pssr_w}. 
Namely, to exactly solve NSM can be beneficial in terms of support recovery 
when only small samples are available. 
On the other hand, 
BFS tends to get faster as $n$ increases. 
This is because 
the condition of NSM instances typically becomes better as $n$ increases, 
which often makes the gap between $F(S)$ and $\low{S}$ smaller; 
i.e., 
$\low{S}$ can accurately estimate $F(S)$, 
and so BFS terminates quickly. 
To conclude, 
given moderate-size NSM instances whose true supports can hardly be recovered with inexact methods, 
BFS can be useful for recovering 
them at the expense of computation time. 
In \Cref{app:additional}, 
we see that, given NSM instances with stronger regularization, 
BFS becomes faster while the PSSR performance deteriorates.

\subsection{Real-world Instances}\label{subsec:real}
We compare BFS and MIP by using NSM instances with quadratic objectives, 
$P(x) = \Lquadratic(Ax) + \frac{1}{2}\|\lambda\|^2$.  
The vector $b\in\R^n$ and matrix $A\in\R^{n\times d}$ 
were obtained from 
two scikit-learn datasets:  
Diabetes and Boston house-price, 
which for simplicity we call Boston. 
We normalized $b$ and each column of $A$ 
so that their $\ell_2$-norm became $1$.  
As in \citep{efron2004least,bertsimas2016best}, 
we considered 
the interaction and square effects of 
the original columns of $A$; 
any square effect 
whose original column had identical non-zeros 
was removed since they are redundant. 
Consequently, we obtained $A\in\R^{n\times d}$ 
with $(n,d)=(442,65)$ and $(506,103)$ 
for the Diabetes and Boston datasets, respectively. 
We let $\lambda=0.001$ and $k=10$. 
We applied BFS and MIP that 
accept various values of objective errors: 
$\Delta=0, 2.0\times10^{-6}, 4.0\times10^{-6}, \dots, 1.0\times10^{-5}$. 
The $\Delta$ value of MIP was controlled 
with the Gurobi parameter, {\tt MIPGapAbs}. 
For comparison, we also applied the inexact methods to the instances, 
whose behavior is independent of $\Delta$ values.   
We evaluated the methods 
in terms of running times, solver calls, 
and objective errors, 
which 
are defined by $P(x)-P(x^*)$ with output solution $x$. 
The solver call of MIP is the number of nodes explored by Gurobi.

\Cref{fig:real} summarizes the results. 
Both BFS and MIP become faster as $\Delta$ increases, 
and BFS is faster than MIP with every $\Delta$ value; 
in Boston instance with $\Delta=0$, 
BFS is more than $300$ times faster than MIP.  
Comparing the results with the two datasets ($d=65$ and $103$), 
we see that BFS is more scalable to large instances than MIP. 
For $\Delta\ge4\times10^{-6}$, 
BFS invoked \solver{} only once. 
Namely, BFS detected that the solutions obtained by a single invocation of \solver{} 
were guaranteed to have at most $\Delta$ errors without examining the descendant nodes. 
In contrast, MIP examined more nodes to obtain the $\Delta$-error guarantees, 
resulting in longer running times than those of BFS. 
This result is consistent with what we mentioned in \Cref{subsec:related}. 
We see that BFS found optimal solutions except for the case 
of Boston instance with $\Delta=2\times10^{-6}$, 
while none of the inexact methods succeeded in exactly solving both instances. 
We remark that the objective error of 
BFS does not always increase with $\Delta$ 
since the priority value, $\low{S}$, is not completely correlated with 
$F(S)$; 
%the true minimum objective value that is achievable in $\desc(S)$;  
i.e., a better solution can be obtained earlier.

\section{CONCLUSION}\label{sec:conclusion}
We proposed a BFS algorithm for NSM with $\ell_2$-regularized convex objective functions. 	
%	The key idea is to prioritize candidate supports appropriately 
%	by leveraging the latest continuous optimization methods. 
%	We also developed pruning and warm-start techniques 
%	that greatly accelerate our BFS. 
%
Experiments confirmed that 
our BFS  existing exact methods are inapplicable, 
%	which inexact methods 
%	which inexact methods 
%	fail to solve, 
and that BFS can run faster than MIP 
with the latest commercial solver, Gurobi 8.1.0. 
%	For future work, 
%	we will develop a new BFS framework  
%	that can be applied to non-convex minimization 
%	with more complex constraints~\citep{hegde2015nearly,jain2016structured}. 

%% The file named.bst is a bibliography style file for BibTeX 0.99c
{
	\bibliographystyle{named}
	\bibliography{mybib}

\begin{thebibliography}{}

\bibitem[\protect\citeauthoryear{Arai \bgroup \em et al.\egroup
  }{2015}]{arai2015optimal}
H.~Arai, C.~Maung, and H.~Schweitzer.
\newblock Optimal column subset selection by {A}-star search.
\newblock In {\em Proceedings of the 29th AAAI Conference on Artificial
  Intelligence}. AAAI Press, 2015.

\bibitem[\protect\citeauthoryear{Avis and Fukuda}{1996}]{avis1996reverse}
D.~Avis and K.~Fukuda.
\newblock Reverse search for enumeration.
\newblock {\em Discrete Appl. Math.}, 65(1):21 -- 46, 1996.

\bibitem[\protect\citeauthoryear{Bertsimas and
  Van~Parys}{2017}]{bertsimas2017sparse}
D.~Bertsimas and B.~Van~Parys.
\newblock Sparse high-dimensional regression: Exact scalable algorithms and
  phase transitions.
\newblock {\em arXiv preprint arXiv:1709.10029}, 2017.

\bibitem[\protect\citeauthoryear{Bertsimas \bgroup \em et al.\egroup
  }{2016}]{bertsimas2016best}
D.~Bertsimas, A.~King, and R.~Mazumder.
\newblock Best subset selection via a modern optimization lens.
\newblock {\em Ann. Statist.}, 44(2):813--852, 2016.

\bibitem[\protect\citeauthoryear{Bhatia \bgroup \em et al.\egroup
  }{2017}]{bhatia2017consistent}
K.~Bhatia, P.~Jain, P.~Kamalaruban, and P.~Kar.
\newblock Consistent robust regression.
\newblock In {\em Advances in Neural Information Processing Systems 30}, pages
  2110--2119. Curran Associates, Inc., 2017.

\bibitem[\protect\citeauthoryear{Blumensath and
  Davies}{2009}]{blumensath2009iterative}
T.~Blumensath and M.~E. Davies.
\newblock Iterative hard thresholding for compressed sensing.
\newblock {\em Appl. Comput. Harmon. Anal.}, 27(3):265--274, 2009.

\bibitem[\protect\citeauthoryear{Bourguignon \bgroup \em et al.\egroup
  }{2016}]{bourguignon2016exact}
S.~Bourguignon, J.~Ninin, H.~Carfantan, and M.~Mongeau.
\newblock Exact sparse approximation problems via mixed-integer programming:
  Formulations and computational performance.
\newblock {\em IEEE Trans. Signal Process.}, 64(6):1405--1419, 2016.

\bibitem[\protect\citeauthoryear{Chen \bgroup \em et al.\egroup
  }{2015}]{chen2015filtered}
W.~Chen, Y.~Chen, and K.~Weinberger.
\newblock Filtered search for submodular maximization with controllable
  approximation bounds.
\newblock In {\em Proceedings of the 18th International Conference on
  Artificial Intelligence and Statistics}, volume~38, pages 156--164. PMLR,
  2015.

\bibitem[\protect\citeauthoryear{Dechter and
  Pearl}{1985}]{dechter1985generalized}
R.~Dechter and J.~Pearl.
\newblock Generalized best-first search strategies and the optimality of {A}*.
\newblock {\em J. ACM}, 32(3):505--536, 1985.

\bibitem[\protect\citeauthoryear{Defazio}{2016}]{defazio2016simple}
A.~Defazio.
\newblock A simple practical accelerated method for finite sums.
\newblock In {\em Advances in Neural Information Processing Systems 29}, pages
  676--684. Curran Associates, Inc., 2016.

\bibitem[\protect\citeauthoryear{Ebendt and
  Drechsler}{2009}]{ebendt2009weighted}
R.~Ebendt and R.~Drechsler.
\newblock Weighted {$\text{A}^*$} search --- unifying view and application.
\newblock {\em Artificial Intelligence}, 173(14):1310 -- 1342, 2009.

\bibitem[\protect\citeauthoryear{Efron \bgroup \em et al.\egroup
  }{2004}]{efron2004least}
B.~Efron, T.~Hastie, I.~Johnstone, and R.~Tibshirani.
\newblock Least angle regression.
\newblock {\em Ann. Statist.}, 32(2):407--499, 2004.

\bibitem[\protect\citeauthoryear{Elenberg \bgroup \em et al.\egroup
  }{2018}]{elenberg2018restricted}
E.~R. Elenberg, R.~Khanna, A.~G. Dimakis, and S.~Negahban.
\newblock Restricted strong convexity implies weak submodularity.
\newblock {\em Ann. Statist.}, 46(6B):3539--3568, 2018.

\bibitem[\protect\citeauthoryear{Foucart}{2011}]{foucart2011hard}
S.~Foucart.
\newblock Hard thresholding pursuit: An algorithm for compressive sensing.
\newblock {\em SIAM J. Optim.}, 49(6):2543--2563, 2011.

\bibitem[\protect\citeauthoryear{Hansen and Zhou}{2007}]{hansen2007anytime}
E.~A. Hansen and R.~Zhou.
\newblock Anytime heuristic search.
\newblock {\em J. Artif. Int. Res.}, 28(1):267--297, 2007.

\bibitem[\protect\citeauthoryear{Hart \bgroup \em et al.\egroup
  }{1968}]{hart1968formal}
P.~E. Hart, N.~J. Nilsson, and B.~Raphael.
\newblock A formal basis for the heuristic determination of minimum cost paths.
\newblock {\em IEEE Trans. Syst. Sci. Cybernet.}, 4(2):100--107, 1968.

\bibitem[\protect\citeauthoryear{Hocking and
  Leslie}{1967}]{hocking1967selection}
R.~R. Hocking and R.~N. Leslie.
\newblock Selection of the best subset in regression analysis.
\newblock {\em Technometrics}, 9(4):531--540, 1967.

\bibitem[\protect\citeauthoryear{Huang \bgroup \em et al.\egroup
  }{2018}]{huang2018constructive}
J.~Huang, Y.~Jiao, Y.~Liu, and X.~Lu.
\newblock A constructive approach to ${L}_0$ penalized regression.
\newblock {\em J. Mach. Learn. Res.}, 19(1):403--439, 2018.

\bibitem[\protect\citeauthoryear{Jain \bgroup \em et al.\egroup
  }{2014}]{jain2014iterative}
P.~Jain, A.~Tewari, and P.~Kar.
\newblock On iterative hard thresholding methods for high-dimensional
  {M}-estimation.
\newblock In {\em Advances in Neural Information Processing Systems 27}, pages
  685--693. Curran Associates, Inc., 2014.

\bibitem[\protect\citeauthoryear{Karahanoglu and
  Erdogan}{2012}]{karahanoglu2012astar}
N.~B. Karahanoglu and H.~Erdogan.
\newblock A* orthogonal matching pursuit: Best-first search for compressed
  sensing signal recovery.
\newblock {\em Digit. Signal Process.}, 22(4):555 -- 568, 2012.

\bibitem[\protect\citeauthoryear{Liu \bgroup \em et al.\egroup
  }{2017}]{liu2017dual}
B.~Liu, X.-T. Yuan, L.~Wang, Q.~Liu, and D.~N. Metaxas.
\newblock Dual iterative hard thresholding: From non-convex sparse minimization
  to non-smooth concave maximization.
\newblock In {\em Proceedings of the 34th International Conference on Machine
  Learning}, volume~70, pages 2179--2187. PMLR, 2017.

\bibitem[\protect\citeauthoryear{Malitsky and Pock}{2018}]{malitsky2018first}
Y.~Malitsky and T.~Pock.
\newblock A first-order primal-dual algorithm with linesearch.
\newblock {\em SIAM J. Optim.}, 28(1):411--432, 2018.

\bibitem[\protect\citeauthoryear{Miyashiro and
  Takano}{2015}]{miyashiro2015mixed}
R.~Miyashiro and Y.~Takano.
\newblock Mixed integer second-order cone programming formulations for variable
  selection in linear regression.
\newblock {\em European J. Oper. Res.}, 247(3):721 -- 731, 2015.

\bibitem[\protect\citeauthoryear{Natarajan}{1995}]{natarajan1995sparse}
B.~K. Natarajan.
\newblock Sparse approximate solutions to linear systems.
\newblock {\em SIAM J. Optim.}, 24(2):227--234, 1995.

\bibitem[\protect\citeauthoryear{Pati \bgroup \em et al.\egroup
  }{1993}]{pati1993orthogonal}
Y.~C. Pati, R.~Rezaiifar, and P.~S. Krishnaprasad.
\newblock Orthogonal matching pursuit: recursive function approximation with
  applications to wavelet decomposition.
\newblock In {\em Proceedings of the 27th Asilomar Conference on Signals,
  Systems and Computers}, pages 40--44 vol.1, 1993.

\bibitem[\protect\citeauthoryear{Pearl}{1984}]{pearl1984heuristics}
J.~Pearl.
\newblock {\em Heuristics: Intelligent Search Strategies for Computer Problem
  Solving}.
\newblock Addison-Wesley Longman Publishing Co., Inc., Boston, MA, USA, 1984.

\bibitem[\protect\citeauthoryear{Sakaue and
  Ishihata}{2018}]{sakaue2018accelerated}
S.~Sakaue and M.~Ishihata.
\newblock Accelerated best-first search with upper-bound computation for
  submodular function maximization.
\newblock In {\em Proceedings of the 32nd AAAI Conference on Artificial
  Intelligence}, 2018.

\bibitem[\protect\citeauthoryear{Sato \bgroup \em et al.\egroup
  }{2016}]{sato2016feature}
T.~Sato, Y.~Takano, R.~Miyashiro, and A.~Yoshise.
\newblock Feature subset selection for logistic regression via mixed integer
  optimization.
\newblock {\em Comput. Optim. Appl.}, 64(3):865--880, 2016.

\bibitem[\protect\citeauthoryear{Shalev-Shwartz and
  Zhang}{2016}]{shalev2016accelerated}
S.~Shalev-Shwartz and T.~Zhang.
\newblock Accelerated proximal stochastic dual coordinate ascent for
  regularized loss minimization.
\newblock {\em Math. Program.}, 155(1):105--145, 2016.

\bibitem[\protect\citeauthoryear{Valenzano \bgroup \em et al.\egroup
  }{2013}]{valenzano2013using}
R.~Valenzano, S.~J. Arfaee, J.~Thayer, R.~Stern, and N.~R. Sturtevant.
\newblock Using alternative suboptimality bounds in heuristic search.
\newblock In {\em Proceedings of the 23rd International Conference on Automated
  Planning and Scheduling}, pages 233--241. AAAI Press, 2013.

\bibitem[\protect\citeauthoryear{Yuan \bgroup \em et al.\egroup
  }{2014}]{yuan2014gradient}
X.~Yuan, P.~Li, and T.~Zhang.
\newblock Gradient hard thresholding pursuit for sparsity-constrained
  optimization.
\newblock In {\em Proceedings of the 31st International Conference on Machine
  Learning}, volume~32, pages 127--135. PMLR, 2014.

\bibitem[\protect\citeauthoryear{Yuan \bgroup \em et al.\egroup
  }{2016}]{yuan2016exact}
X.~Yuan, P.~Li, and T.~Zhang.
\newblock Exact recovery of hard thresholding pursuit.
\newblock In {\em Advances in Neural Information Processing Systems 29}, pages
  3558--3566. Curran Associates, Inc., 2016.

\end{thebibliography}
}

\clearpage
\appendix

%%%%%%%%%% Merge with supplemental materials %%%%%%%%%%
\clearpage
%\onecolumn
%\pagebreak
%\widetext
\begin{center}
	\textbf{\LARGE Appendix}
\end{center}
%%%%%%%%%% Merge with supplemental materials %%%%%%%%%%
%%%%%%%%%% Prefix a "S" to all equations, figures, tables and reset the counter %%%%%%%%%%
%	\setcounter{equation}{0}
%	\setcounter{section}{0}
%	\setcounter{figure}{0}
%	\setcounter{table}{0}
%	\setcounter{page}{1}
%	\makeatletter
%%%%%%%%%% Prefix a "S" to all equations, figures, tables and reset the counter %%%%%%%%%%
%	\renewcommand{\thesection}{S\arabic{section}}
%	\renewcommand{\thefigure}{S\arabic{figure}}

\newcommand{\us}{{\bar u}}
\newcommand{\jp}{{j^\prime}}
\newcommand{\jhat}{{\hat j}}
\newcommand{\vtl}{{\tilde v}}
\newcommand{\js}{{j_{{\tt start}}}}
\newcommand{\je}{{j_{{\tt end}}}}
\newcommand{\jso}{{j^*_{{\tt start}}}}
\newcommand{\jeo}{{j^*_{{\tt end}}}}
\newcommand{\gmin}{{g_{\min}}}
\newcommand{\ol}{{\tt Endpoints}}
\newcommand{\el}{{\tt Examined}}

\section{Computing Proximal Operator in \Cref{alg:pdal}}
\label{app:prox}

We first introduce Moreau's identity: 
\[
x = \alpha \prox_{\alpha^{-1} f}(\alpha^{-1} x) + \prox_{\alpha f^*}(x)
\] 
for any $\alpha>0$, $x\in\R^m$, and (proper closed) 
convex function $f:\R^m\to\R$.  
With this equality, 
the computation of $\prox_{\tau_{t-1}L^*}(\beta^{t-1} - \tau_{t-1}Ay^{t-1})$ in Step~\ref{step:pdal_beta_update} of \Cref{alg:pdal} 
can be reduced to the computation of the proximal operator for convex loss function $L(\cdot)$, 
which can be performed efficiently with various $L(\cdot)$. 

We next see how to compute 
$\prox_{\rho_t\tau_{t} (\frac{1}{2\lambda} \|\cdot\|^2_{k-s,2})^*}(\bar y^{t}_{S_>})$ 
in Step~\ref{step:prox_topk2}. 
Thanks to Moreau's identity, it can be written as  
\[
\bar y^{t}_{S_>} - \rho_t\tau_t \prox_{\frac{1}{2\lambda\rho_t\tau_{t}} \|\cdot\|^2_{k-s,2}} 
\left(\frac{1}{\rho_t\tau_t} \bar y^{t}_{S_>} \right).
\]
Therefore, if we can compute the proximal operator for the top-$(k-s)$ $\ell_2$-norm,
then we can perform Step~\ref{step:prox_topk2}.   
Below we show that it can be computed in $O(d)$ time.

\paragraph{Proximal-operator Computation for the Top-$k$ $\ell_2$-norm.}
Given any positive integer $k$, $d$ ($k\le d$), 
vector $v\in\R^d$, 
and parameter $\mu>0$, we show how to efficiently compute 
\begin{align}
\prox_{\frac{\mu}{2}\topknorm{\cdot}{k}^2}(v)
&=
\argmin_{x\in\R^d}
\left\{
\frac{\mu}{2}\topknorm{x}{k}^2 + \frac{1}{2}\|x-v\|^2
\right\}
\\
&=
\argmin_{x\in\R^d}
\left\{
\mu \topknorm{x}{k}^2 + \|x-v\|^2
\right\}.  
\end{align}
Let $\sign(v)\in\{-1,1\}^d$ denote a vector  
whose $i$th entry is $1$ if $v_i\ge0$ and $-1$ otherwise. 
We have 
$\sign(v)\odot v = (|v_1|,\dots,|v_d|)^\top \eqqcolon |v|$, 
where $\odot$ is the Hadamard (element-wise) product. 
We define $\sigma_{|v|}:\dset\to\dset$ as a permutation 
that rearranges the entries 
of $|v|$ in a non-increasing order. 
We abuse the notation and 
take $\sigma_{|v|}(x)\in\R^d$ to be a permutated vector 
for any given $x\in\R^d$; 
note that $\sigma_{|v|}(|v|)$ is non-increasing.  
We also define $\sigma_{|v|}^{-1}$ as the inverse permutation 
that satisfies $x=\sigma_{|v|}^{-1}(\sigma_{|v|}(x))$ for any $x\in\R^d$. 
Let $u\coloneqq\sigma_{|v|}(|v|)$ ,
which is non-negative and non-increasing.  
Note that, once we obtain 
\begin{equation}
\xtl
\coloneqq
\argmin_{x\in\R^d}
\left\{
\mu \topknorm{x}{k}^2 + \|x-u\|^2
\right\}, \label{prob:prox} 
\end{equation} 
then the desired solution, 
$\prox_{\frac{\mu}{2}\topknorm{\cdot}{k}^2}(v)$, 
can be computed as
$\sign(v)\odot\sigma_{|v|}^{-1}(\xtl)$. 
Therefore, below we discuss how to compute $\xtl$. 

We define $f(x)\coloneqq \mu\topknorm{x}{k}+\|x-u\|^2$. 
Since $u_1\ge\dots\ge u_d\ge0$, we have $\xtl_1\ge\dots\ge\xtl_d\ge0$; 
otherwise $\xtl$ is not a minimizer. 
We let $\js\in[k]$ and $\je\in\{k,\dots,d\}$ be the smallest and largest indices 
such that 
$\xtl_\js=\xtl_\je=\xtl_k$; 
i.e., it holds that  
\[
\xtl_1\ge\dots>\xtl_\js=\dots=\xtl_k=\dots=\xtl_\je>\dots\ge\xtl_d.	
\] 
We define
$\us_i\coloneqq \frac{u_i}{1+\mu}$ for $i\in[k]$. 
If $\js=\je=k$, 
we can readily obtain  
\begin{align}
\xtl_i =
\begin{cases}
\us_i  & \text{for $i=1,\dots,k$},
\\
u_i  & \text{for $i=k+1,\dots,d$}. 
\end{cases}
\end{align}
Note that this case occurs iff $\us_k\ge u_{k+1}$; 
in this case, $\xtl$ can be obtained as above. 
We then consider the case $\js<\je$. 
Since $f$ is convex  
and $\xtl$ is a minimizer, 
we have $0\in\partial f(\xtl)$,  
which implies 
\begin{align}
\xtl_i =
\begin{cases}
\us_i  & \text{for $i=1,\dots,\js-1$},
\\
u_i  & \text{for $i=\je+1,\dots,d$}.
\end{cases}
\end{align}
Namely, 
$\xtl_1,\dots,\xtl_{\js-1}$ 
and 
$\xtl_{\je+1},\dots,\xtl_d$ 
can readily be obtained. 
Below we discuss how to compute $\js$, $\je$, 
and $\xi\coloneqq \xtl_\js=\dots=\xtl_\je$.  
Since $\xtl$ is a minimizer of $f$, 
our aim is to  
find an optimal triplet, 
$(\js,\je,\xi)\in[k]\times\{k,\dots,d\}\times\R$, 
that 
satisfies 
\begin{align}
\xi\in[u_{\je+1},\us_{\js-1}]
\label{eq:zeta_constraint}
\end{align}
and 
minimizes 
\begin{align}\label{prob:kyfan} 
g(\js, \je, \xi) 
\coloneqq{}
(1+\mu)\sum_{i=\js}^k (\xi - \us_i)^2
+
\sum_{i=k+1}^{\je} (\xi - u_i)^2,  
\end{align} 
where we regard $\us_0=+\infty$ and $u_{d+1}=0$, 
and we take the second term on the RHS to be $0$ if $\je=k$.  
Note that, 
once $(\js, \je)$ is fixed, 
computing optimal $\xi$ reduces to 
a one-dimensional quadratic minimization problem 
with constraint~\eqref{eq:zeta_constraint}, 
whose solution $\xi$ 
can be written as follows: 
\begin{align}
\xi 
=
\min
\{
\us_{\js-1},
\max
\{
u_{\je+1}, 
\tilde{\xi}
\}
\}, 
\end{align}
where
\[
\tilde{\xi} 
=
\frac{\sum_{i=\js}^{\je} u_i}{\mu(k-\js+1)+\je-\js+1}.
\]
In what follows, we discuss how to find $(\js,\je)$ that constitutes 
an optimal triplet.

We first show that triplet $(\js,\je,\xi)$ 
that satisfies~\eqref{eq:zeta_constraint} is sub-optimal 
if $\xi< \us_\js$ holds.   
In this case, we have 
\[
g(\js+1,\je,\xi) \le g(\js,\je,\xi)
\]  
and 
triple $(\js+1,\je,\xi)$ satisfies constraint~\eqref{eq:zeta_constraint}, 
i.e., 
\[
\xi\in[u_{\je+1}, \us_{\js}],
\] 
since $\xi<\us_{\js}$. 
Namely, $(\js+1,\je,\xi)$ 
is feasible and achieves at least as small $g$ value as $(\js,\je,\xi)$. 
Below
we focus on the case where $\us_\js\le\xi$ holds. 

\begin{algorithm}[tb]
	{\fontsize{9pt}{10pt}\selectfont	
		\caption{Computation of $\prox_{\frac{\mu}{2}\topknorm{\cdot}{k}^2}(v)$}\label{alg:kyfan}
		\begin{algorithmic}[1]
			\State $u\gets \sigma_{|v|}(|v|)$ 
			\State $\us_i\gets u_i/(1+\mu)$ for $i\in[k]$ 
			\If{$\us_k \ge u_{k+1}$}
			\State
			$
			\xtl\gets(\us_1,\dots,\us_k, u_{k+1},\dots, u_d)^\top
			$
			\State {\bf return} 
			$
			\sign(v)\odot\sigma_{|v|}^{-1}(\xtl)
			$
			\EndIf
			\State $\jhat\gets k$, 
			$g_{\min}\gets+\infty$, 
			and 
			$\el\gets\emptyset$  
			\For{$\js = 1,\dots, k$}
			\State $\ol\gets\{j\in[d]\relmid{} u_j > \us_\js \ \text{and}\ j\ge\jhat \}$
			\For{$\je\in\ol$}
			\State
			$\tilde{\xi}\gets
			\frac{\sum_{i=\js}^{\je} u_i}{\mu(k-\js+1)+\je-\js+1}$
			\State $\xi\gets\min
			\{
			\us_{\js-1},
			\max
			\{
			u_{\je+1}, 
			\tilde{\xi}
			\}
			\}$  
			\If{$g(\js,\je,\xi)< g_{\min}$}
			\State 
			$(\jso,\jeo,\xi^*)\gets(\js,\je,\xi)$
			\State
			$\gmin\gets g(\jso,\jeo,\xi^*)$
			\EndIf		 
			\EndFor
			\State $\el\gets\el\cup\ol$
			\If{$\el\neq\emptyset$}
			\State $\jhat\gets\max\el$  
			\EndIf
			\EndFor
			\State
			$\xtl\gets(\us_1,\dots,\us_{\jso-1},\xi^*,\dots,\xi^*,u_{\jeo+1},\dots,u_d)^\top$
			\State {\bf return} 
			$\sign(v)\odot\sigma_{|v|}^{-1}(\xtl)$
		\end{algorithmic}
	}
\end{algorithm}

We then prove that 
triple $(\js, \je, \xi)$ 
satisfying \eqref{eq:zeta_constraint} and 
$\us_{\js}\ge u_{\je}$ is sub-optimal. 
In this case, 
we have 
\[
g(\js,\je-1,\xi)\le g(\js,\je,\xi)
\]  
and 
triple $(\js, \je-1, \xi)$ satisfies constraint~\eqref{eq:zeta_constraint}, 
i.e., 
\[
\xi\in[u_{\je}, \us_{\js-1}], 
\] 
since $u_\je \le \us_{\js} \le \xi$. 
Namely, $(\js,\je-1,\xi)$ 
is feasible and achieves at least as small $g$ value as $(\js,\je,\xi)$. 	
Therefore, to find an optimal triplet, 
we only need to examine $(\js,\je,\xi)$ 
that satisfies constraint~\eqref{eq:zeta_constraint} 
and 
\begin{align}
u_\je>\us_{\js}.
\label{ineq:uus}
\end{align}

We fix $\js\in[k]$ 
and define  
\[
\jhat\coloneqq \max\{j\in\{k,\dots,d\} \relmid{} u_j>\us_{\js-1} \}. 
\]  
Then, 
\begin{align}
\je \ge \jhat
\quad 
\text{(i.e., $u_\je \le u_\jhat$)}
\label{ineq:jej}
\end{align} 
must hold for the following reason: 
If $\je+1\le\jhat$ holds, we have 
$
u_{\je+1}\ge u_{\jhat} > \us_{\js-1}  
$, 
which means no $\xi$ satisfies constraint~\eqref{eq:zeta_constraint}. 

Taking \eqref{ineq:uus} and \eqref{ineq:jej} into account, 
once $\js\in[k]$ is fixed, 
endpoint $\je$ to be examined satisfies 
\begin{align}
u_\je > \us_{\js} 
\quad 
\text{and} 
\quad
\je
\ge
\jhat = \max\{j\in\{k,\dots,d\} \relmid{} u_j>\us_{\js-1} \}. 
\end{align}
Therefore, 
by 
examining $\js=1,\dots,k$ sequentially 
and 
maintaining 
$\el=\{j\in\{k,\dots,d\} \relmid{} u_j > \us_\js \}$ 
as in \Cref{alg:kyfan}, 
we can find an optimal triplet $(\jso,\jeo,\xi^*)$, 
with which we can obtain 
$\prox_{\frac{\mu}{2}\topknorm{\cdot}{k}^2}$.  

We examine the complexity of \Cref{alg:kyfan}. 
Let $\ol_i$ be the list of endpoints constructed in the $i$th iteration for $i\in[k]$. 
Since $\ol_{i}$ and $\ol_{i+1}$ have at most one common element 
and $\bigcup_{i\in[k]}\ol_i$ includes at most $d-k+1$ elements, 
we have $\sum_{i=1}^{k}|\ol_i|\le d-k+1 + (k-1)=d$. 
Namely, \Cref{alg:kyfan} examines at most $d$ candidate triplets, 
hence \Cref{alg:kyfan} runs in $O(d)$ time.

\begin{algorithm}[tb]
	{\fontsize{9pt}{10pt}\selectfont
		\caption{SGA for computing $\low{S}$ and $\sol{S}$}\label{alg:diht}
		\begin{algorithmic}[1]
			\State Initialize $\beta^{0}$ and $\eta_0$. 
			\State Fix $\epsilon>0$. 
			\Comment{$\epsilon=10^{-5}$ in the experiments}   
			\For{$t=1,2,\dots$}
			\State 
			Compute a supergradient $g^{t-1}\in\partial D(\beta^{t-1};S)$.
			\State $\eta_{t}\gets2\times\eta_{t-1}$
			\Loop 
			\Comment{Backtracking of step-size $\eta_t$}
			\State 
			$
			\beta^{t}\gets
			\PF(\beta^{t-1}+\eta_{t}g^{t-1}) 
			$ 
			\If{$D(\beta^{t};S)\ge D(\beta^{t-1};S)$}
			\State {\bf break}
			\EndIf
			\State $\eta_{t} \gets 0.5\times\eta_{t}$
			\EndLoop
			\If{$(D(\beta^{t};S)-D(\beta^{t-1};S))/\pmin\le\epsilon$}
			\State 
			\begin{tabular}{@{}l}
				$x^{t}_S \gets -\frac{1}{\lambda} A_S^\top \beta^{t}$, 
				\quad 
				$x^{t}_{S_\myle \setminus S} \gets 0$, 
				\quad 
				and 
				\\
				$x^{t}_{S_\myg} \gets -\frac{1}{\lambda} \mathcal T_{k-s} (A_{S_\myg}^\top \beta^{t})$
			\end{tabular} 
			\State
			$\low{S}\gets D(\beta^{t};S)$
			\State
			$\sol{S}\gets\argmin_{\supp(x)\subseteq \supp(x^{t})} \P(x)$
			\State {\bf return} $\low{S}$, $\sol{S}$ 
			\EndIf
			\EndFor
		\end{algorithmic}
	}
\end{algorithm}

\section{Comparison of Supergradient Ascent and Primal-dual Algorithm with Linesearch}
\label{app:sga_pdal}

As mentioned in~\Cref{subsection:diht}, 
we can use the supergradient ascent (SGA) for solving
\begin{align}
\maximize_{\beta \in \R^n} \quad D(\beta; S)
\end{align}
instead of PDAL. 
We first describe the details of SGA based on \citep{liu2017dual}, 
and then we experimentally compare two BFS algorithms 
that use SGA and PDAL as their subroutines.

\subsection{Details of SGA}

Let $\efdom\coloneqq\{\beta\in\R^n\relmid{} D(\beta;S)>-\infty\}$ be the effective domain of $D(\beta;S)$ 
and $\PF(\cdot)$ be the Euclidean projection operator onto 
$\efdom$. 
If $\partial L^*(\beta)\subseteq \R^n$ is the super-differential of $L^*(\beta)$, 
the super-differential of $D(\beta;S)$ is given by 
\[
\partial D(\beta;S)
=
\{A \tilde x(\beta;S) - \tilde g \relmid{} 
\tilde g \in\partial L^*(\beta)  \}.
\]
With these definitions, 
the SGA procedure for computing $\low{S}$ and $\sol{S}$ 
can be described as in Algorithm~\ref{alg:diht}. 
%, which we may use to perform Step~\ref{step:solver_comp} of \Cref{alg:oracle}. 
We can use the warm-start and pruning techniques as in \Cref{subsec:acceleration}, 
but the details of the warm-start technique for initializing $\beta^{0}$ and $\eta_{0}$ 
are slightly different from those of PDAL as explained below. 
When executing SGA  with $S=\emptyset$ at the beginning of BFS, 
we set $\beta^{0}\gets 0$ and $\eta_{0}\gets 1$. 
We now suppose that $\Sparent\in\V$ is popped from $\minheap$ 
and that we are about to compute $\low{S}$ and $\sol{S}$, where $(\Sparent,S)\in E$. 
Let $\beta_\Sparent$ be a dual solution obtained in the last iteration of SGA executed for $\Sparent$; 
i.e., $\low{\Sparent}=D(\beta_\Sparent;\Sparent)$.
When executing SGA to maximize $D(\beta;S)$, we let $\beta^{0}\gets\beta_\Sparent$. 
Furthermore, 
we set $\eta^{0}\gets\eta_\Sparent$, 
where 
$\eta_\Sparent$ is a step-size used to obtain $\beta^{1}$ in SGA executed for $\Sparent$.

\begin{figure}[tb]
	\centering
	\begin{minipage}{0.235\textwidth}
		\includegraphics[width=1.0\linewidth]{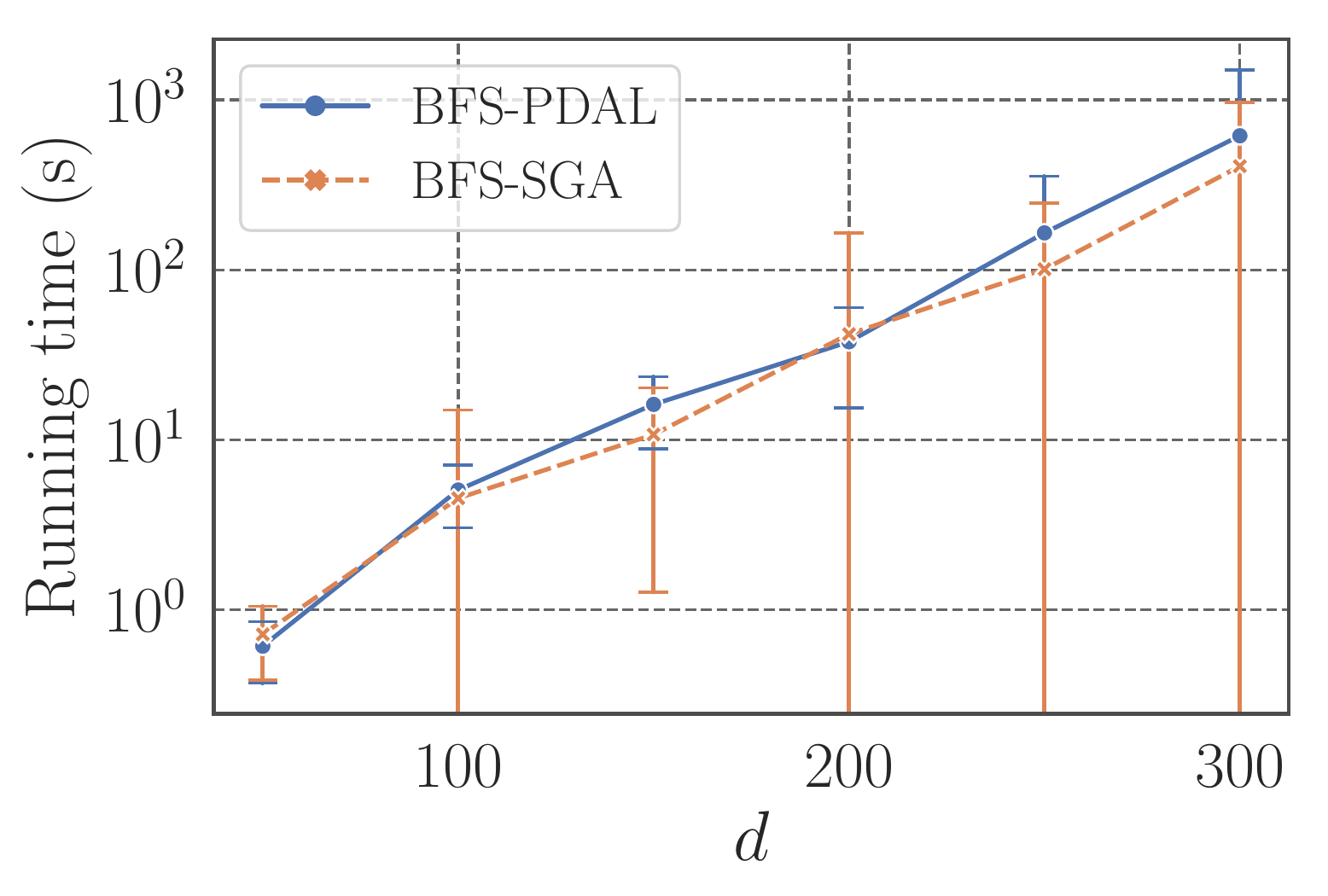}
		\subcaption{Running Time}
	\end{minipage}
	\begin{minipage}{0.235\textwidth}
		\includegraphics[width=1.0\linewidth]{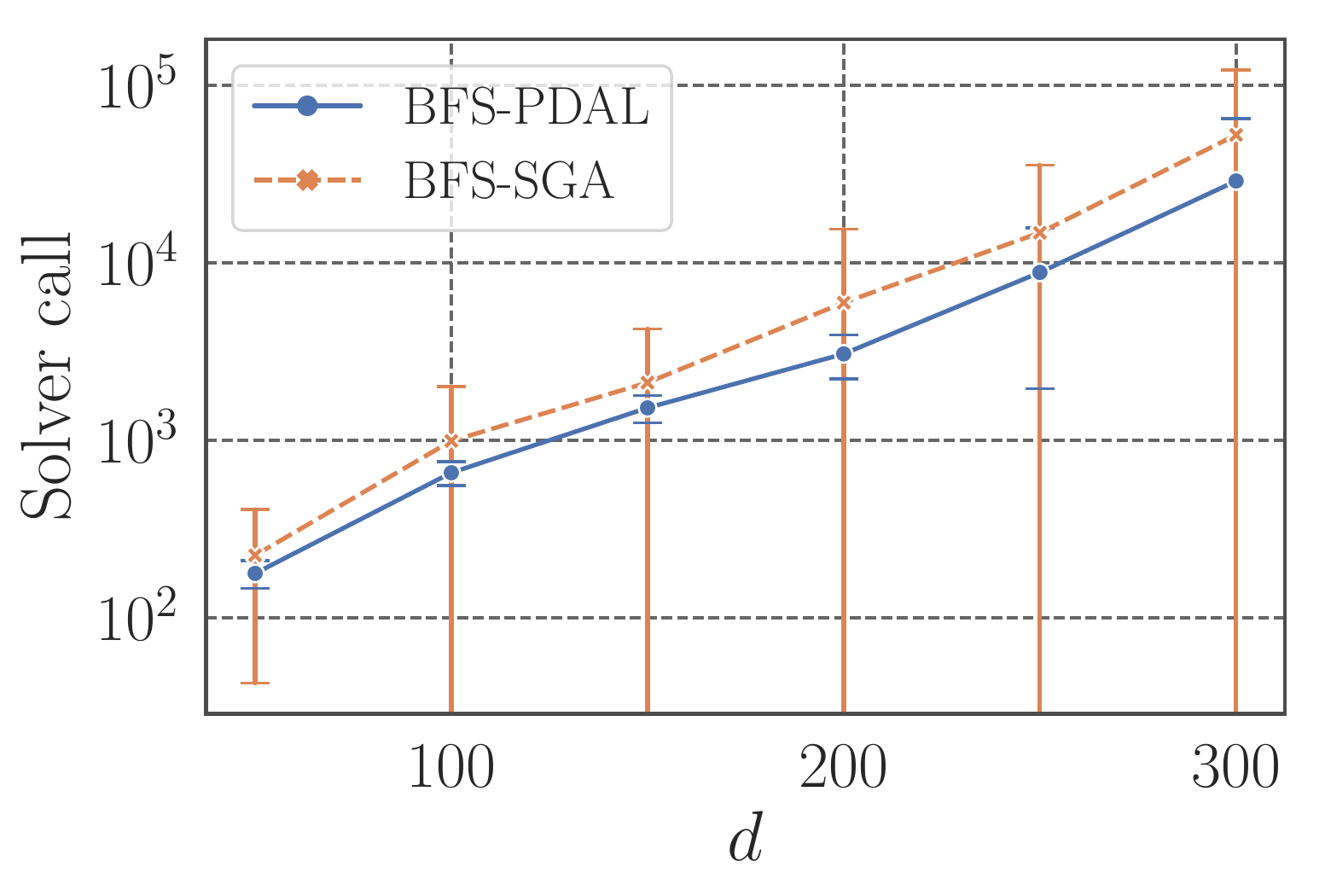}
		\subcaption{Solver Call}
	\end{minipage}
	%\begin{tabular}[t]{p{0.45\textwidth}}
	%\end{tabular}
	\caption{Running Times and Solver Calls of BFS-SGA and BFS-PDAL.} 
	\label{fig:syn_sga_pdal_plot}
\end{figure}	

\subsection{Experimental Comparison}
We experimentally compare two BFS algorithms with SGA and PDAL, 
which we call BFS-SGA and BFS-PDAL, respectively. 	
The experimental setting used here is 
the same as the Huber instance in \Cref{subsubsec:cost}. 
We observed the running times and solver calls of 
BFS-SGA and BFS-PDAL. 
As with BFS-PDAL, 
BFS-SGA used warm-start and pruning techniques. 
%that are similar to those described in \Cref{subsec:acceleration}. 
%Below we present the details of SGA (see also~\citep{liu2017dual}). 

\Cref{fig:syn_sga_pdal_plot} shows the results, 
where each curve and error bar indicate the mean and 
standard deviation calculated over $100$ random instances. 
While the running times of the two methods are almost the same, 
BFS-PDAL is more efficient in terms of solver calls. 
This result implies that $\low{S}$ values computed by PDAL and SGA are 
different even though $D(\beta;S)$ is concave; 
in fact, due to the non-smoothness of $D(\beta;S)$, 
SGA sometimes fails to maximize $D(\beta;S)$. 
As a result, 
BFS-SGA tends to require more solver calls than BFS-PDAL on average. 
Note that, 
while the computation costs of SGA and PDAL are polynomial 
in $d$, 
the number of solver calls can increase exponentially in $k$; 
i.e., 
it is more important to reduce the number of solver calls  
than to reduce the running time of subroutines (SGA and PDAL). 
To conclude, 
BFS-PDAL is expected to be more scalable to larger instances than BFS-SGA, 
which motivates us to employ PDAL.

\begin{figure}[t]
	\centering
	%\begin{tabular}[t]{p{0.45\textwidth}}
	%\includegraphics[width=1.0\linewidth]{./figures/ablation_plot.pdf}
	%\end{tabular}
	\begin{minipage}{0.235\textwidth}
		\includegraphics[width=1.0\linewidth]{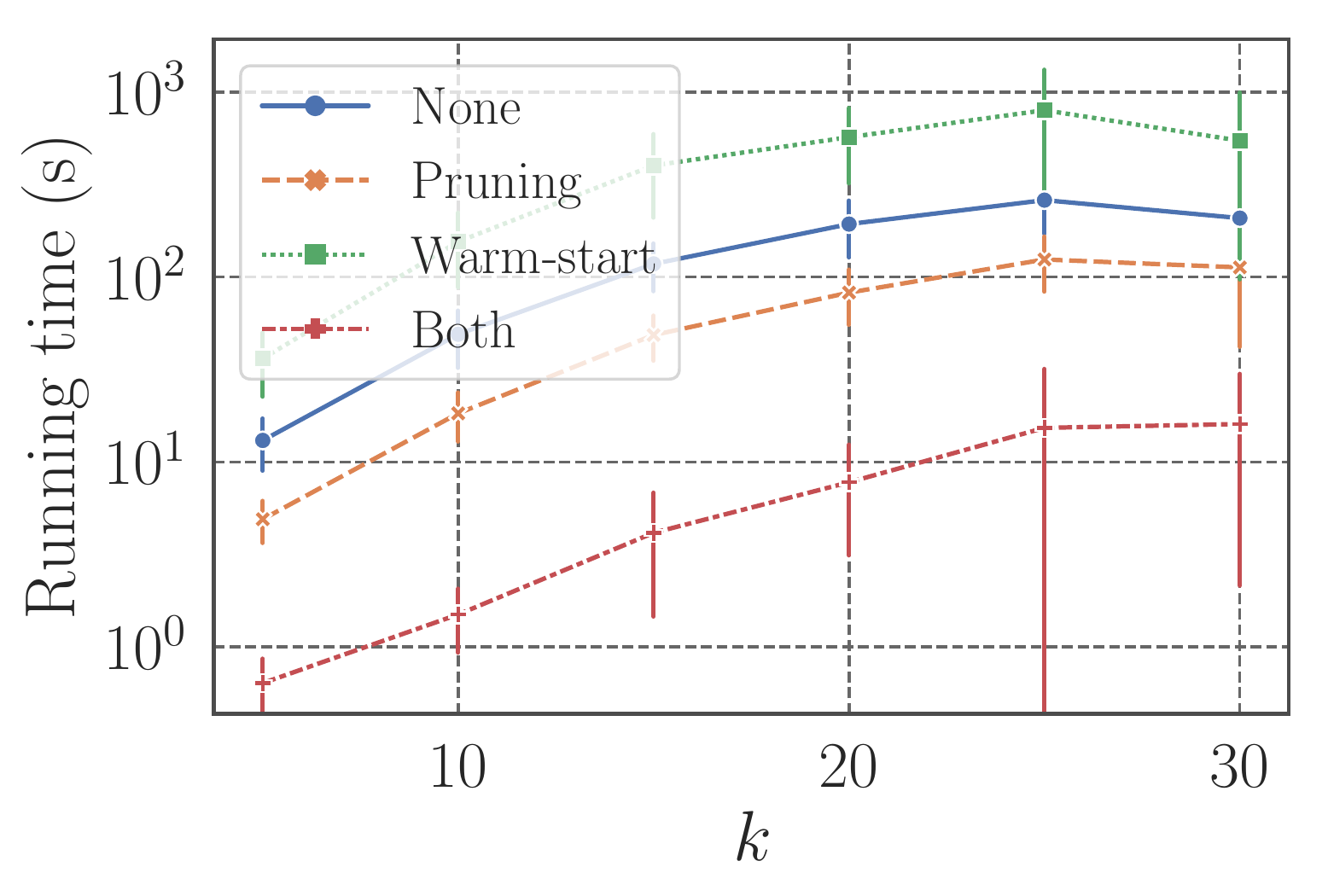}
		\subcaption{Running Time}
	\end{minipage}
	\begin{minipage}{0.235\textwidth}
		\includegraphics[width=1.0\linewidth]{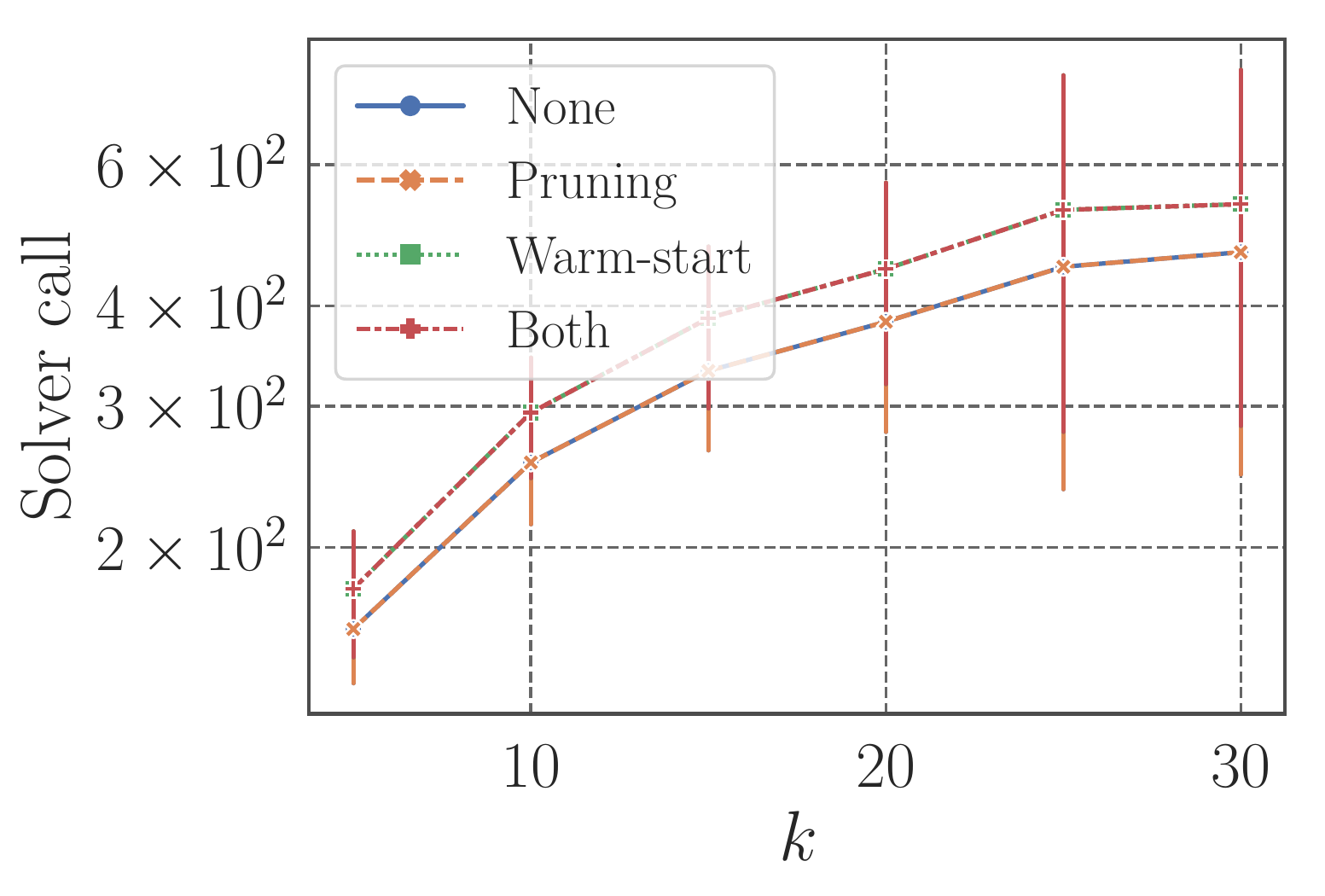}
		\subcaption{Solver Call}
	\end{minipage}
	\caption{
		Running Times and Solver Calls of 
		BFS with and without Pruning and Warm-start Techniques.  	
	} 
	\label{fig:ablation_plot}
\end{figure}

\section{Ablation Study}\label{app:ablation}
We experimentally study the degree to which the pruning and warm-start techniques 
speed up BFS. 
%BFS alone, with pruning, with warm-start, and with both are denoted by 
%BFS, BFS-P, BFS-W, BFS-PW, respectively. 
We used the Huber instances (\Cref{subsec:synthetic}) 
with $d=50$, $k=5,10,\dots,30$, and $n=\floor{10k\log d}$; 
for each $k$ value we generated $100$ random instances.  

\Cref{fig:ablation_plot} presents 
running times and solver calls. 
Each value 
and error bar are mean and standard deviation calculated over 
$100$ instances. 
We see that the warm-start technique alone does not always accelerate BFS, 
but the combination of pruning and warm-start 
greatly reduces the running time. 
This is because, 
if a good solution is available thanks to warm-start 
at the beginning of \solver{}, 
then it can be force-quitted quickly via the pruning procedure. 
We also see that, 
while the size of the state-space tree increases exponentially 
in $k$, 
the running time and solver call grow sub-linearly in 
$k$ in the semi-log plots, 
which implies that the search space is effectively reduced  
thanks to our prioritization method with \solver. 

\section{Additional Experimental Results}
\label{app:additional}

We examine how the PSSR performances and running times 
change with stronger regularization. 
The experimental settings used here are almost the same as those of 
\Cref{subsubsec:pssr}. 
The only difference is the $\lambda$ value: 
We let $\lambda=0.01$ and $0.002$ for Huber and logistic instances, 
respectively. 

\Cref{a_fig:pssr_plot} presents the PSSR values and running times. 
Relative to the results shown in \Cref{subsubsec:pssr}, 
BFS became faster 
and 
the PSSR performance gap between BFS and the inexact methods became smaller. 
The reason for this result is as follows: 
When strongly regularized, 
objective functions become convex more strongly. 
This typically reduces the gap, $F(S) - \low{S}$, and so BFS terminates more quickly. 
Furthermore, it becomes easier to solve NSM instances exactly with inexact methods. 
Namely, it tends to be easy to exactly solve NSM instances with strong regularization. 
On the other hand, due to the over regularization, 
optimal solutions to such NSM instances often fail to recover the true support, 
hence the PSSR performance of BFS deteriorates; 
consequently, the gap between BFS and inexact methods became smaller. 
To conclude, 
if we are to achieve high support recovery performance, 
we need to solve NSM instances with moderate regularization, 
which is often hard for inexact methods as implied by the experimental results in \Cref{subsubsec:pssr}. 
This observation emphasizes the utility of our BFS, 
which is empirically efficient enough for exactly solving moderate-size NSM instances.

\begin{figure}[tb]
	\centering
	\begin{minipage}[t]{0.235\textwidth}
		\includegraphics[width=1.0\linewidth]{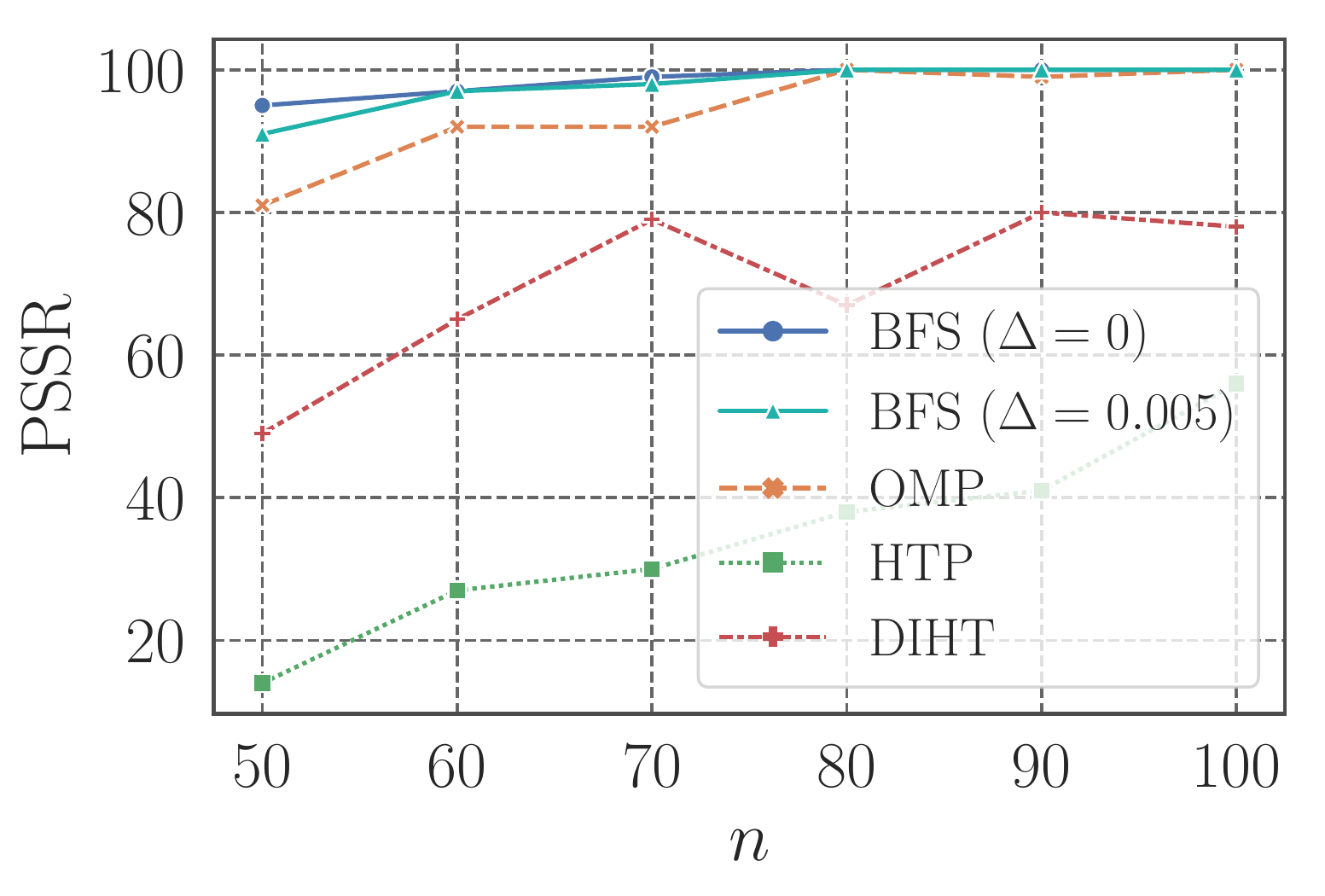}
		\subcaption{Huber, PSSR}
		\label{fig:huber_pssr_s}
	\end{minipage}
	\begin{minipage}[t]{0.235\textwidth}
		\includegraphics[width=1.0\linewidth]{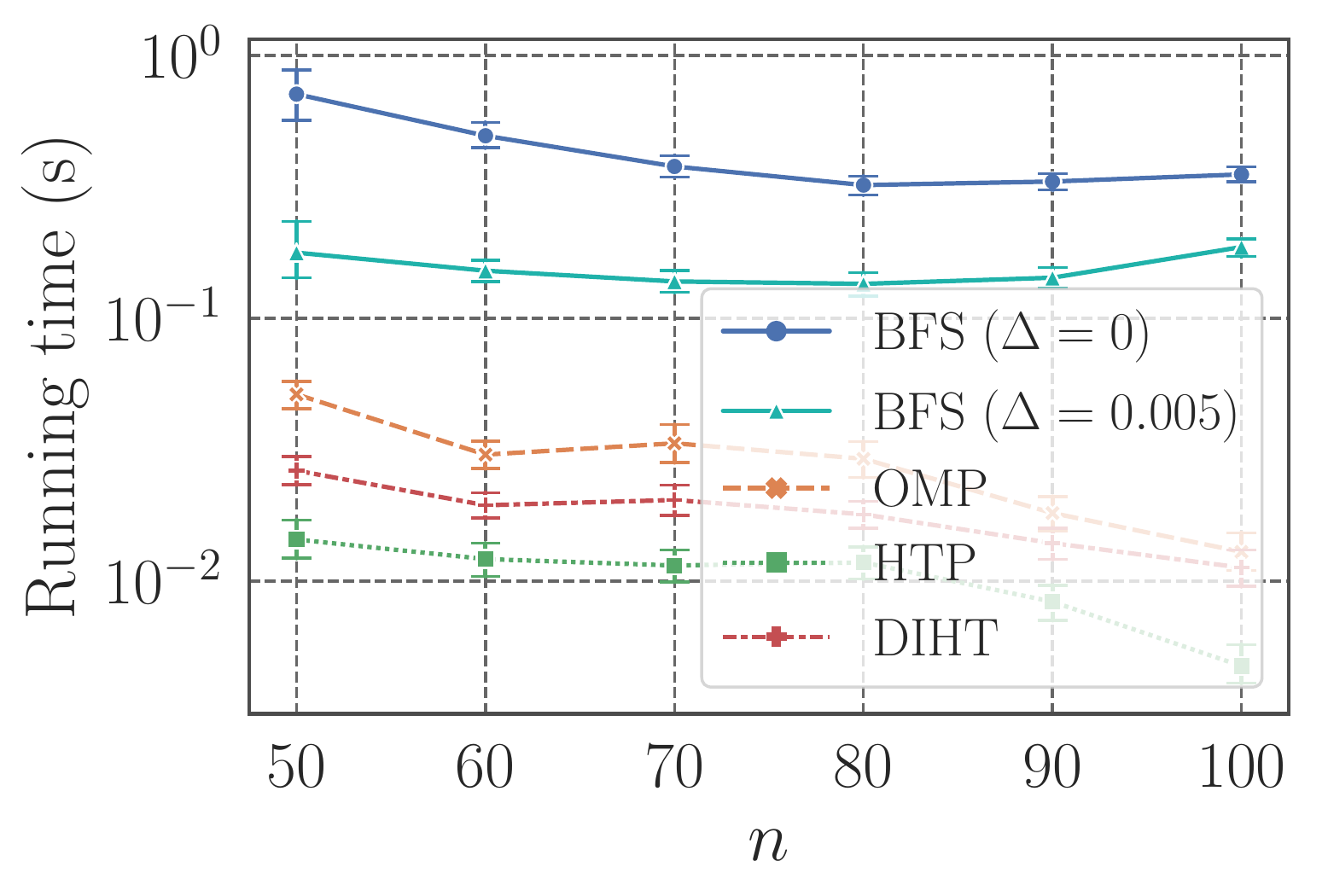}
		\subcaption{Huber, Running Time}
		\label{fig:huber_time_s}
	\end{minipage}
	\begin{minipage}[t]{0.235\textwidth}
		\includegraphics[width=1.0\linewidth]{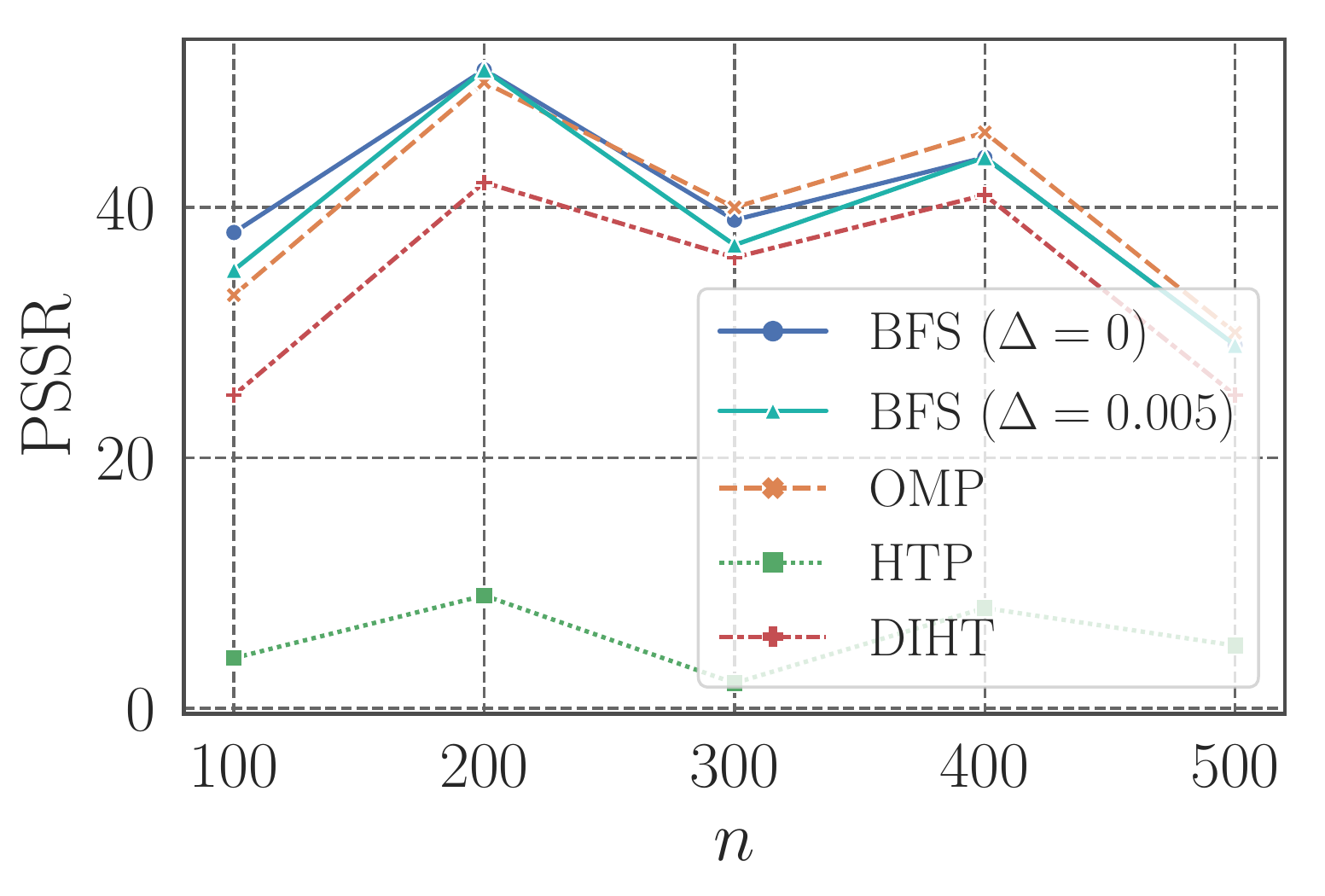}
		\subcaption{Logistic, PSSR}
		\label{fig:logistic_pssr_s}
	\end{minipage}
	\begin{minipage}[t]{0.235\textwidth}
		\includegraphics[width=1.0\linewidth]{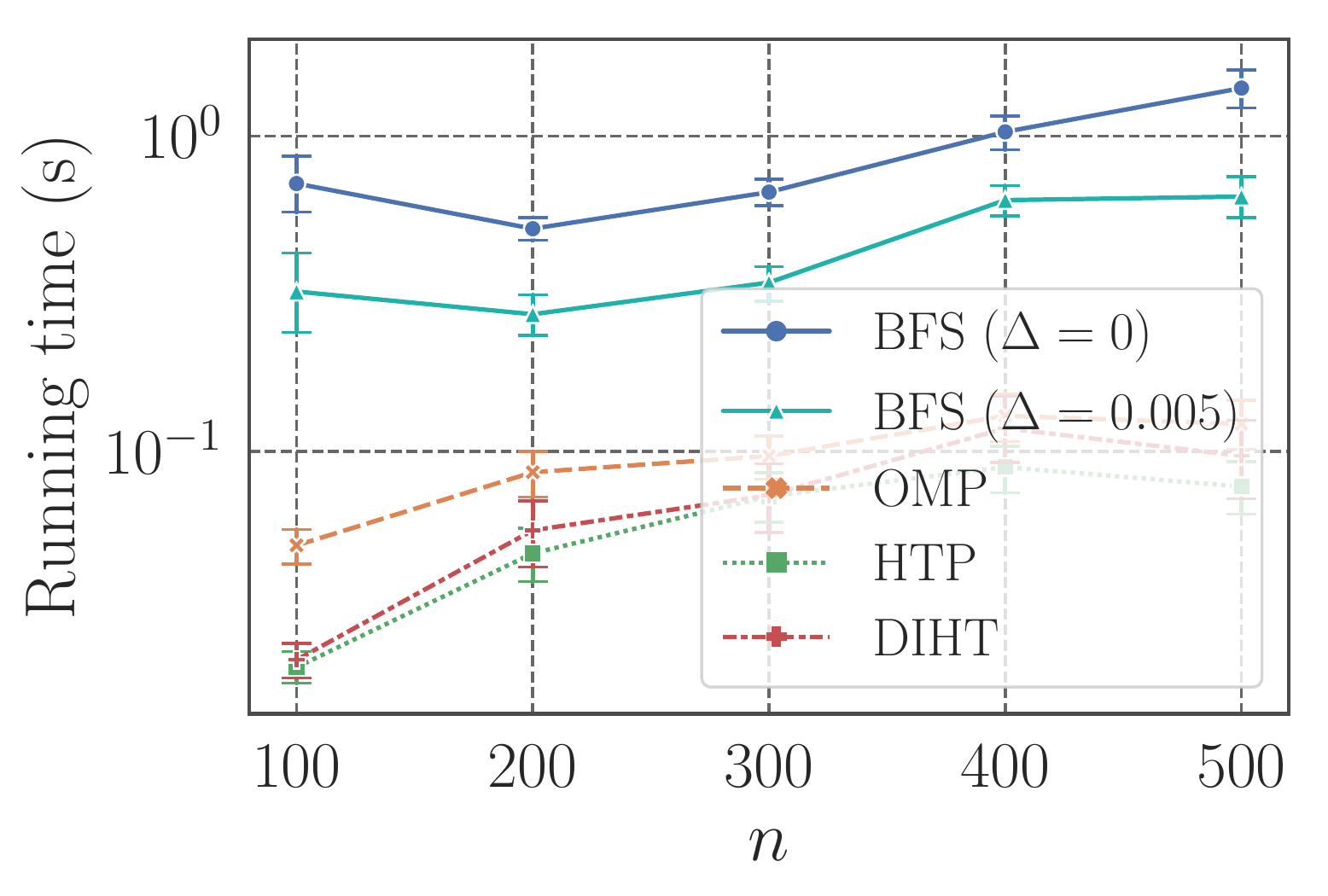}
		\subcaption{Logistic, Running Time}
		\label{fig:logistic_time_s}
	\end{minipage}
	\caption{
		PSSR and Running Times.
	} 
	\label{a_fig:pssr_plot}
\end{figure}

\end{document}